\documentclass{svproc}

\usepackage{microtype}
\usepackage{listings}
\usepackage{subcaption}
\usepackage{multicol}
\usepackage{xcolor}
\usepackage{graphicx}
\usepackage{hyperref}

\usepackage{bold-extra}
\usepackage{amssymb}

\newcommand{\comment}[1]{}

\newcommand{\highlight}[1]{%
  \colorbox{gray!50}{$\displaystyle#1$}}

\lstdefinelanguage{Scala}%
{morekeywords={abstract,case,catch,char,class,%
    def,else,extends,final,%
    if,import,%
    match,module,new,null,object,override,package,private,protected,%
    public,return,super,this,throw,trait,try,type,val,var,with,implicit,%
    macro,sealed,%
  },%
  sensitive,%
  morecomment=[l]//,%
  morecomment=[s]{/*}{*/},%
  morestring=[b]",%
  morestring=[b]',%
  showstringspaces=false%
}[keywords,comments,strings]%

\begin{document}
\mainmatter              
%
\title{Observable atomic consistency for CvRDTs}
\titlerunning{Observable atomic consistency for CvRDTs (Extended Version)}
\author{Xin Zhao \and Philipp Haller}

\authorrunning{X. Zhao and P. Haller} 

\tocauthor{Xin Zhao and Philipp Haller}

\institute{KTH Royal Institute of Technology, Stockholm, Sweden\\
\email{\{xizhao, phaller\}@kth.se}}




\maketitle              

\begin{abstract}
The development of distributed systems requires developers to balance the need for consistency, availability, and partition tolerance. Conflict-free replicated data types (CRDTs) are widely used in eventually consistent systems to reduce concurrency control. However, due to the lack of consistent non-monotonic operations, the usage of CRDTs can be difficult.

In this paper, we propose a new consistency protocol, the observable atomic consistency protocol (OACP). OACP enables a principled relaxation of strong consistency to improve performance in specific scenarios. OACP combines the advantages of mergeable data types, specifically, convergent replicated data types, and reliable total order broadcast to provide on-demand strong consistency. By providing observable atomic consistency, OACP avoids the anomalies of related protocols.

We provide a distributed, cluster-enabled implementation of OACP based on Akka, a widely-used actor-based middleware. Our experimental evaluation shows that OACP can reduce the coordination between replicas compared to other protocols providing atomic consistency in several benchmarks. Our results also suggest that OACP gains availability through mergeable data types and provides acceptable latency for achieving strong consistency.

\end{abstract}

\section{Introduction}
Conflict-free replicated data types (CRDTs)~\cite{MarcS11} are widely used in industrial distributed systems such as Riak~\cite{BrownCME14,RiakDTSource} and Cassandra. They are objects which can be updated concurrently without consensus and eventually converge to the same state if all updates are executed by all replicas eventually. Thus, they can provide high availability and scalability for replicated shared data. However, the main challenges are the fact that they only provide eventual consistency and they necessarily provide only a restricted set of operations. In particular, they do not support consistent non-monotonic operations; for example, read operations may return outdated values before the replicas have converged. This makes the usage of CRDTs difficult.

There are two principal approaches of CRDTs, one is operation-based CRDTs and the other is state-based CRDTs. Operation-based CRDTs are also called CmRDTs, and they send only the update operations which are commutative. State-based CRDTs are also called CvRDTs, and they propagate the whole state to other replicas which can be merged by a commutative function. There are many specifications for CRDT counters, sets, and registers.

In this paper, we focus on CvRDTs and we address these challenges by extending CvRDTs with on-demand strong consistency. We devise a novel observable atomic consistency protocol (OACP) which avoids some anomalies of other recently proposed consistency models.

\subsubsection{Contributions.}
This paper makes the following contributions:
\begin{enumerate}
\item We introduce the observable atomic consistency (OAC) model which
  enables a novel extension of CvRDTs with non-monotonic
  operations. Thus, we lift a major limitation of CvRDTs, and
  significantly simplify programming with CvRDTs. We also provide a
  precise and formal definition of the OAC model.

\item We prove that systems providing observable atomic consistency
  are state convergent. The paper summarizes our definitions and
  results, while our companion technical report~\cite{ZhaoH18}
  contains the complete proofs.

\item We provide the observable atomic consistency protocol (OACP)
  which guarantees observable atomic consistency.

\item We provide a distributed implementation of OACP which is enabled
  to run on clusters. Our implementation is based on the widely-used
  Akka~\cite{Akka} actor-based middleware. The system is available
  open-source on GitHub.\footnote{See
    \url{https://github.com/CynthiaZ92/OACP}}

\item We provide an experimental evaluation of OACP, including latency, throughput and coordination. Our evaluation shows that OACP benefits more when there are more commutative operations and can reduce the number of exchanged protocol messages compared to the baseline protocol. Using the case study of a Twitter-like microblogging service, we experimentally evaluate optimization of OACP, which we call O$^2$ACP.
\end{enumerate}

The rest of the paper is organized as follows. Section~\ref{sec:overview} illustrates how application programmers use OACP. In Section~\ref{sec:consistency} we formalize the observable atomic consistency (OAC) model, and prove that systems providing OAC are state convergent. Section~\ref{sec:oacp} explains the observable atomic consistency protocol (OACP), and compares it to the Global Sequence Protocol~\cite{BurckhardtLPF15} (GSP). In Section~\ref{sec:eval} we present a performance evaluation of an actor-based implementation of OACP using microbenchmarks as well as a Twitter-like application. Section~\ref{sec:rel-work} discusses related work, and Section~\ref{sec:conclusion} concludes.

\section{Overview}\label{sec:overview}
We provide a usage overview of the OACP system from the perspective of a programmer to provide a user-friendly interface towards distributed application development. First, we introduce the system structure, and then we demonstrate how to use the provided API through an example.

\paragraph{System structure.}
Figure~\ref{graph:system} shows the system structure. There are three layers from bottom to top: storage, distributed protocol, and application. We first focus on the top layer and discuss how applications interface with the protocol.

\begin{figure}[!ht]
\centering
\includegraphics[width=1.50in]{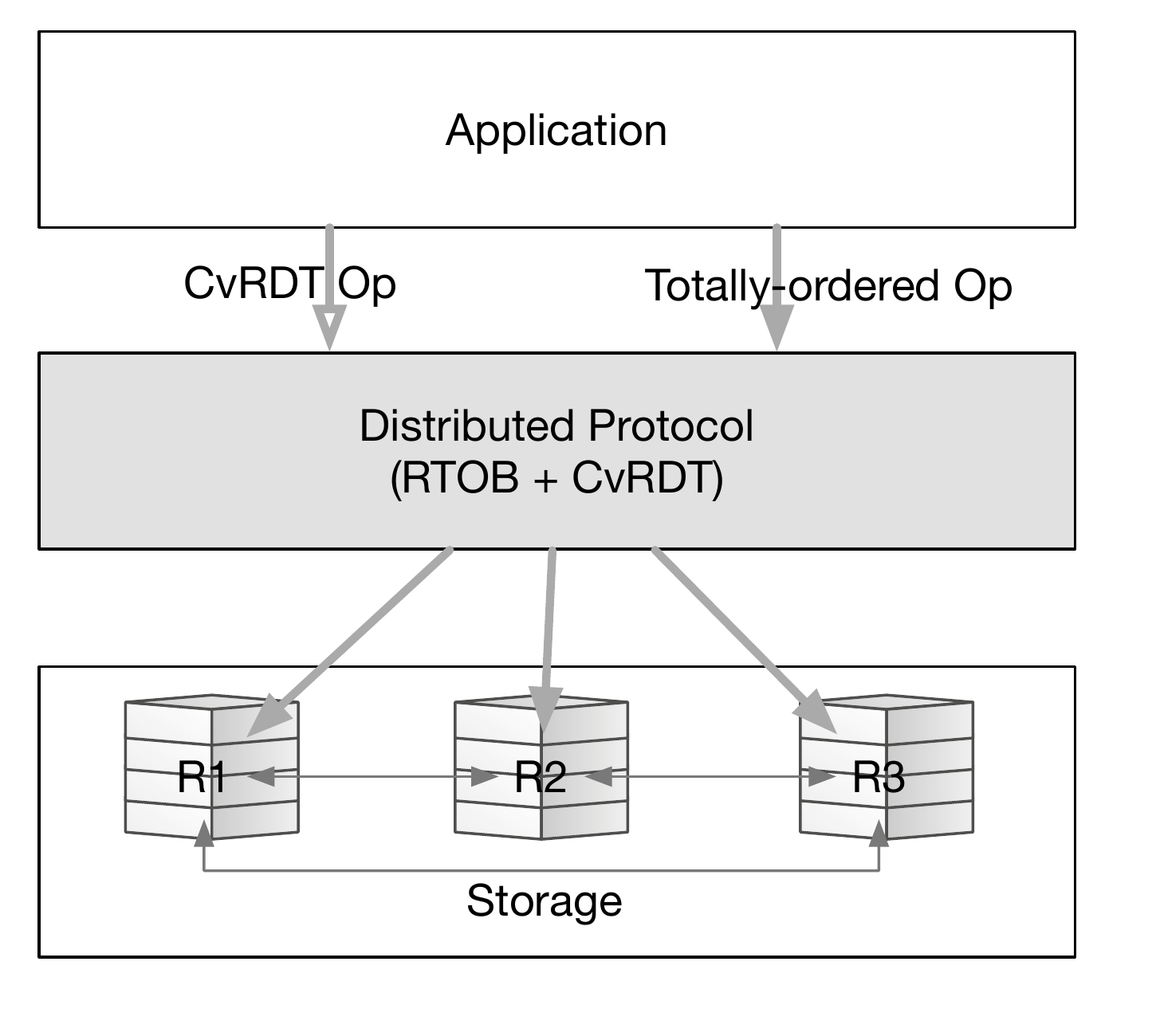}
\caption{High-level view of OACP}\label{graph:system}
\end{figure}

The operations submitted by the application are divided into two categories: CvRDT operations (CvOps) and totally-ordered operations (TOps). CvOps are commutative. TOps are supported by reliable total order broadcast (RTOB)~\cite{DefSchUrb04}, so that their ordering is preserved across the entire system.

The submission of a TOp causes all replicas to atomically (a) synchronize their convergent states, and (b) lock their convergent states. A replica with locked convergent state buffers CvOps until the original TOp has been committed. Moreover, at the point when a TOp is executed, all replicas are guaranteed to have consistent convergent states. Thus, submitting a TOp ensures the consistency of all replicas, including their convergent states.

\paragraph{User API.}
The user-facing API consists of the following three parts.
\begin{itemize}
    \item trait \verb|CvRDT[T]|
    \item CvOp (CvValue, CvUpdate)
    \item TOp (TUpdate)
\end{itemize}
\noindent
The \verb|CvRDT| trait\footnote{A trait in Scala can be thought of as similar to a Java interface enabling a safe form of multiple inheritances.} exposes the implementation of CvRDT to developers which allows them to define their own CvRDT types and operations which are provided by the trait such as initialisation, add, remove and merge.

The last two methods provide access to use the OACP, allowing developers to decide whether the operation should be operated as CvRDTs or ordered messages. CvOps are used when one wants to gain benefits from high availability since they are directly sent to the closest available server. While TOps are good options if the developer wants to make all the replicas reach the same state for some essential operations. These two message types give developers more flexibility when they want to achieve certain consistency as well as the performance of the implementation. The following is an example showing how to use the interface.

\paragraph{Resettable counter.}

The grow-only counter (GCounter) is one of the most basic counters which is widely used, e.g., in real-time analytics or in distributed gaming. It is a CvRDT which only supports increment and merge operations. However, when a system employs a GCounter to achieve eventual consistency, often a special ``reset'' operation is needed for resetting the counter to its default initial state. In other words, we need an ``observed-reset''~\cite{BaqueroAL16} operation for GCounter to go back to the bottom state, which means when ``reset'' is invoked, all effects of the GCounter observed in different replicas should be equivalently reset. However, the standard GCounter cannot solve this problem; this limitation is also well-known in the popular Riak DT implementation.

Thus, we need an implementation of a resettable counter to make sure that all the replicas are reset at the same time. A straightforward solution is to define ``reset'' as a totally-ordered operation, leveraging the property of TOps.

Figure~\ref{list:gcounter} shows a GCounter definition in Scala. We extend the \verb|CvRDT| trait to have an instance of GCounter which supports operations such as \verb|incr| and \verb|merge|. Each replica in the cluster is assigned an ID; this enables each \verb|GCounter| instance to increment locally. When merging the states of two GCounters, we take the maximum counter of each index. The \verb|compare| method is used to express the partial order relationship between different \verb|GCounter|s.

\begin{figure}[h!]
  \begin{lstlisting}[numbers=left, numberstyle=\scriptsize\color{gray}\ttfamily, backgroundcolor = \color{white}]
trait CvRDT[T] {
  def myID(): Int
  def merge(other: T): Unit
  def compare(other: T): Boolean
}

abstract class GCounter extends CvRDT[GCounter] {
  val p: Array[Int] = Array.ofDim[Int](3)
  def incr(): Unit = {
    val id = myID()
    p(id) = p(id) + 1
  }
  def merge(other: GCounter): Unit =
    for (i <- 0 until p.length)
      p(i) = math.max(p(i), other.p(i))
  def compare(other: GCounter): Boolean =
    (0 until p.length).forall(i => p(i) <= other.p(i))
}
  \end{lstlisting}\caption{GCounter in Scala.}\label{list:gcounter}
\end{figure}

Then the message handlers need to be defined on the client side so that OACP system can recognize the behavior. We extend the client actor by OACP protocol to connect application layer with OACP protocol. The developer only needs to define which message to send using the following ``Akka-style'' message handler. When the CounterClient receives a certain message, it behaves according to the user's definition.

\begin{lstlisting}[numbers=left, numberstyle=\scriptsize\color{gray}\ttfamily, backgroundcolor = \color{white}]
class CounterClient extends Protocol[GCounter] {
  val CounterClientBehavior: Receive = {
    case Incr  => self forward CvOp("incr")
    case Reset => self forward TOp("Reset")
    ...
  }
  override def receive = CounterClientBehavior.orElse(super.receive)
  ...
}
\end{lstlisting}
\noindent
For example, case \verb|Incr| handles a CvOp which is forwarded to the protocol layer.

\section{Observable atomic consistency}\label{sec:consistency}

We now formalize the consistency model of OACP.

\begin{definition}[CvT order]
Given a set of operations $U = C \cup T$ where $C \cap T = \varnothing$, a CvT order is a partial order $O=(U,\prec)$ with the following restrictions:
\begin{itemize}
  \item $\forall u, v \in T$ such that $u \neq v.~u \prec v \lor v \prec u$
  \item $\forall p \in C$, $u \in T.~p \prec u \lor u \prec p$
\end{itemize}
\end{definition}

According to the transitivity of the partial order, we could derive that $\forall l, m, n \in U$ such that $l \prec m, m \prec n$. $l \prec n$.

\begin{definition}[Cv-set]

  Given a set of operations $U = C \cup T$ where $C \cap T =
  \varnothing$, a Cv-set $C_i$ is a set of $C$
  operations with the restriction that:

  \begin{itemize}
    \item $\forall p, q \in C_i \Rightarrow p \nprec q \land q \nprec p$
    \item $\forall p \in C \setminus C_i.\exists q \in C_i$ such that $p \prec q \lor q \prec p$
  \end{itemize}
\end{definition}

\begin{figure}[h!]
        \centering
        \includegraphics[width=0.5\textwidth]{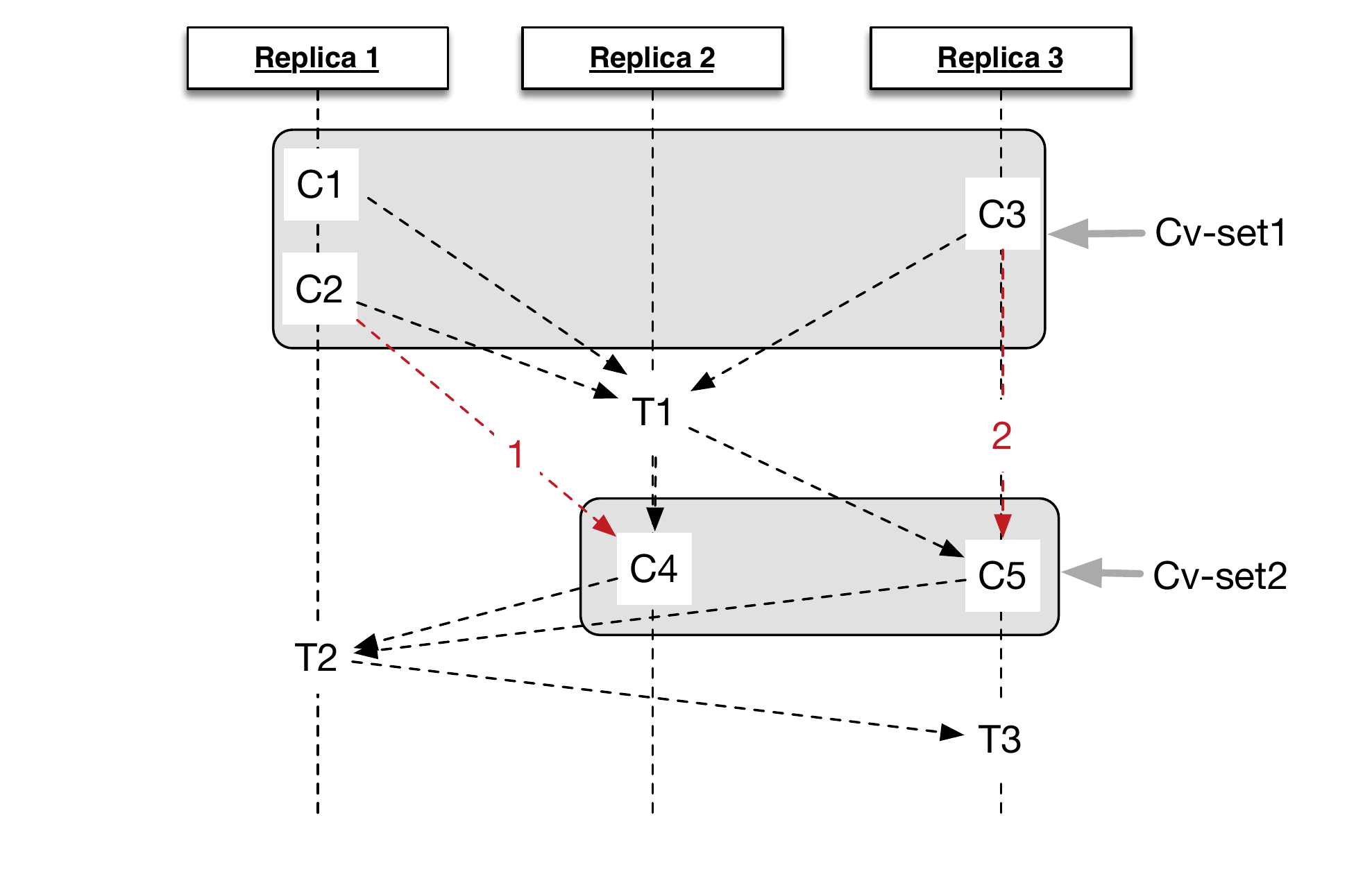}
        \caption{CvT order of operations.}\label{fig:mnorder}
\end{figure}

In Figure~\ref{fig:mnorder}, the partial order is labeled with black dash arrow, while we could derive line 1 and 2 (red dash arrow) from the transitivity of $\prec$. There are two different Cv-sets in this situation while the CvRDT updates inside have no partial order relationship.

In CvT order, only the operations in one same Cv-set could happen concurrently. Now we consider the operations on different sites. Suppose each site $i$ executes a linear extension compatible with the CvT order. Then, the replicated system with $n$ sites provides local atomic consistency, which is defined as follows.

\begin{definition}[Local atomic consistency (LAC)]
  A replicated system provides local atomic consistency (LAC) if each site $i$ applies operations according to a linear extension of the CvT order.
\end{definition}

The three replicas in Figure~\ref{fig:mnorder} could have different linear extensions of the CvT order. However, despite possible re-orderings, LAC guarantees state convergence.


\begin{definition}[State convergence]
  A LAC system is state convergent if all linear extensions of the underlying CvT order $O$ reach the same state $S$.
\end{definition}

\begin{theorem}\label{lac:stateconvergent}
  Given a CvT order, if all operations in each $Cv$-set are commutative, then any LAC system is state convergent.
\end{theorem}

In order to give a complete proof for theorem~\ref{lac:stateconvergent}, we first introduce the following lemmas and their proofs.
\begin{lemma}\label{lemma 1}
Given a legal serialization $O_i = (U, <_i)$ of CvT order $O=(U, \prec)$, if $\exists u, v \in U$ such that $u <_i v$ and $u \nprec v$, then $\exists ! C_{i} \in U$ such that $u \in C_{i} \wedge v \in C_{i}$.
\end{lemma}
\begin{proof}
Let's consider two operations in the following cases and find which situation satifies the requriment.\\
\textbf{Case 1:} if $u \in T$ and $v \in T$, $u <_i v$ then according to the definition, $u \prec v$. \\
\textbf{Case 2:} if $u \in T$ and $v \in C$, $u <_i v$, there must $\exists C_{i}$ such that $v \in C_{i}$, since $u \prec (\forall c|c \in C_{i}$, $u \prec v$.\\
\textbf{Case 3:} if $u \in C$ and $v \in T$, $u <_i v$, using the same argument from Case 2, we have $u \prec v$.\\
\textbf{Case 4:} if $u \in C$ and $v \in C$, $u \in C_{i} \wedge v \in C_{j} \wedge C_{i} \cap C_{j} = \varnothing$, $u <_i v$ so that $C_{i} \prec C_{j}$, then we have $u \prec v$.\\
\textbf{Case 5:} if $u \in C_{i}$ and $v \in C_{i}$, $u \in C_{i} \wedge v \in C_{i}$, we know from the definition that $u \nprec v$.\\
Consider the above five cases, only case 5 can satisfy $u <_i v \wedge u\nprec v$, so the lemma is correct.
\end{proof}

\begin{lemma}\label{lemma 2}
Assume $O_i = (U,<_i)$ and $O_j = (U,<_j)$ are both legal serializations of CvT order $O = (U,\prec)$ that are identical except for two adjacent operations $a$ and $b$ such that $a <_i b$ and $b <_j a$ and that all operations $c \in C$ are commutive inside each single $Cv$-set. Then $S(O_i) = S(O_j)$.
\end{lemma}

\begin{proof}[Proof for lemma]
Let $P$ and $S$ be the greatest common prefix and suffix of $O_i$ and $O_j$, then $S_{ab}=S(P)+a+b$, $S_{ba}=S(P)+b+a$.

From lemma~\ref{lemma 1}, we know that only the CvRDT operations inside the same $Cv$-set can be the two adjacent operations $a$ and $b$. The operation in one $Cv$-set commute, so $S_{ab} = S_{ba}$. And we could also derive that $S_{ab} + S(Q) = S_{ba} + S(Q)$. Since $S_{ab} + S(Q) = S_(O_i)$ and $S_{ba} + S(Q) = S_(O_j)$ then we have $S(O_i) = S(O_j)$
\end{proof}

Now we come back to the proof of Theorem 1.1.
\begin{proof}[Proof for theorem]
Considering when:

(1) $O_i = O_j$. Then we have $S(O_i) = S(O_j)$.
{}
(2) $O_i \neq O_j$, then there exists at least one pair of adjacent operations such as $u$ and $v$ that $u <_i v$ and $v <_j u$. From lemma~\ref{lemma 2} we know if we swap the sequence of $u$ and $v$ in $O_i$, we could get $O_{i+1} = O_i$. If $O_{i+1} \neq O_j$, then we continue to find adjacent operation pairs in $O_{i+1}$ until we finally construct a sequence of operations which is equal to $O_j$.

Then we proved that any LAC system is state convergent.
\end{proof}

\begin{definition}[Observable atomic consistency]\label{oacpdefinition}
  A replicated system provides observable atomic consistency (OAC) if
  it provides local atomic consistency and for all $p \in C$, $u \in
  T$.~(a) $p \prec_{PO} u$ in program order (of some client) implies
  that $p \prec u$ in the CvT order, and (b) $u \prec_{PO} p$ in
  program order (of some client) implies that $u \prec p$ in the CvT
  order.
\end{definition}

The constraints for OAC guarantee that the system state is consistent with the order of operations on each client side. Since OAC is a particular case for LAC, the state convergence also holds for OAC, and we have the following corollary:

\begin{corollary}
  Given a CvT order, if all operations in each $Cv$-set are commutative, then any OAC system is state convergent.
\end{corollary}

\paragraph{Discussion.}
Observable atomic consistency (OAC) is a novel consistency model. It guarantees that invocations of totally-ordered operations (TOps) establish consistency across all replicas. Between two subsequent TOps there can be invocations of convergent operations (CvOps) which can be re-ordered as long as they happen between the two original TOps. We now compare our consistency model with previous notions.

\paragraph*{Comparison to RedBlue Consistency.}

A closely related consistency model is RedBlue consistency~\cite{LiPCGPR12}. Red operations are the ones which need to be totally ordered while the blue ones can commute globally. Two operations can commute means that the order of the operations does not affect the final result. In order to adapt an existing system to a RedBlue system, {\em shadow operations} with commutativity property need to be created and then adjust the application to use these shadow operations. The shadow operations can ``introduce some surprising anomalies to a user experience''~\cite{LiPCGPR12}. In other words, the final system state might not take all messages into account that were received by the different replicas. Thus, RedBlue consistency can only maintain a local view of the system. In OAC, the CvRDT updates can only commute inside a specific scope, namely, between two totally-ordered operations. This restricts the flexibility of CvRDT updates, but at the same time, it provides a consistent view of the state. Importantly, the final system state always matches the state resulting from a linear extension of the original partial order of the operations. An example comparing OAC and RedBlue consistency is included in our companion technical report~\cite{ZhaoH18}.

Here is an example for comparison. There are Alice in the EU and Bob in the US, performing some operations on the same account. Suppose the sequence of bank operations arriving at any of the replicas is as follows:
Alice(deposit(20)) $\rightarrow$ Bob(accrueinterest()) $\rightarrow$ Bob(withdraw(60)) $\rightarrow$ Bob(deposit(10)) $\rightarrow$ Alice(withdraw(70)).

RedBlue consistency requires the definition of withdrawAck' and withdrawFail' as shadow operations for an original ``withdraw'' operation. Then, withdrawAck' is marked as a red operation (i.e., totally ordered).

In RedBlue consistency, when Alice wants to execute withdraw(70) (blue frame in Figure~\ref{graph:bank}(A)), there is an immediate response withdrawFail'() since it is a blue operation. In the local view of Alice, the deposit(10) from Bob is not visible yet. However, in the assumed arriving sequence of operations, deposit(10) from Bob was already received by one of the replicas; thus, in principle the withdraw could have succeeded.

\begin{figure}[h!]
        \includegraphics[width=\textwidth]{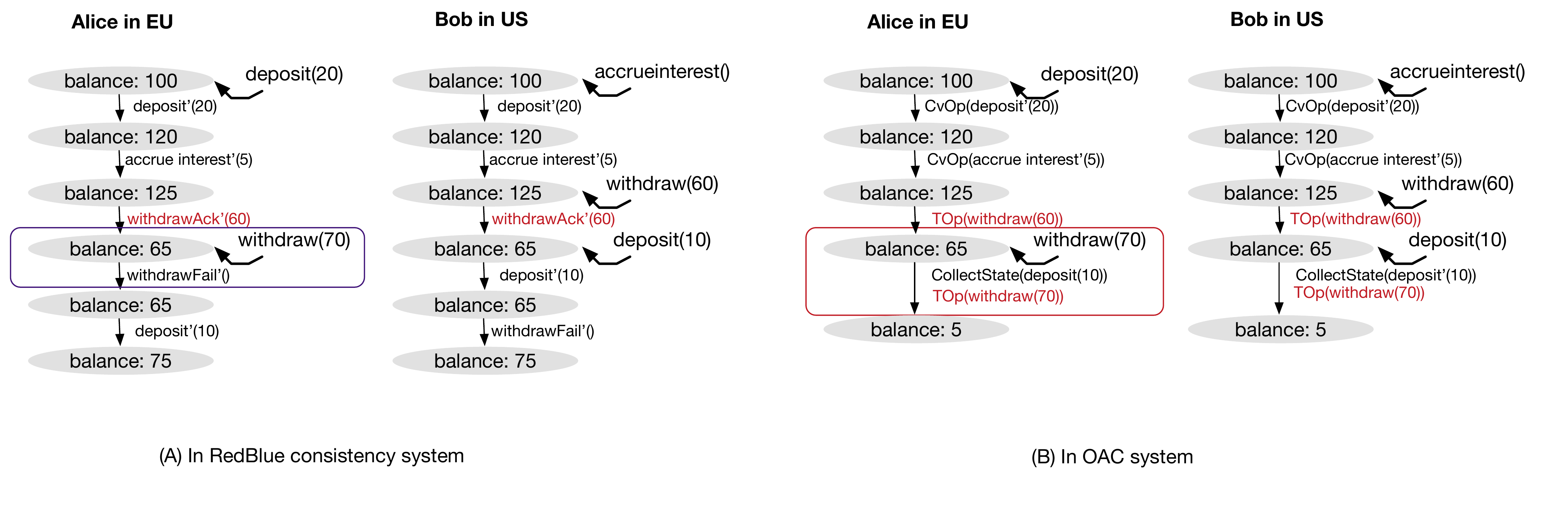}
        \caption{Comparison between RedBlue consistency and OAC}
        \label{graph:bank}
\end{figure}

Compared with RedBlue consistency, in OAC system, the evolution of the system state looks as shown in Figure~\ref{graph:bank}(B). we only need to define which operations should be CvOps and which should be TOps. Meanwhile, when Alice wants to execute withdraw(70), since it is a TOp, the system will make all the replicas reach the same state to make sure Alice has a consistent view of the system (red frame in Figure~\ref{graph:bank}(B)).

\section{Observable Atomic Consistency Protocol (OACP)}
After the definition of OAC in Section~\ref{sec:oac}, now we denote the semantics of the protocol using a so-called data model together with the pseudocode.

\subsection{Data model}
In this part, we describe the data models that are used in the protocol. Consider all the clients and servers as a set of $objects$. In general, we use notation $Op$ to define the operations that are performed on objects. $Op^*$ represents a sequence of operations.

Operation types are classified according to different names. Thus, the function $Eval: Op^* \times Op \rightarrow Val$ presents the result when applying the sequence of operations on one object.

We define abstract type CvUpdate and TUpdate stand for CvRDT updates and totally-ordered updates, $Value$ for result from reading the log, $Read$ for an ordered read operation. Then we could give a data model using the function $OACPVal: Read \times \{CvUpdate, TUpdate\}^* \rightarrow Value$. The protocol will be explained in a client-server topology.

In the following sections, we use $Log$ to record all the sequence of operations and states. The unit of Log is called $Entry$, which contains the object information, number of next log index, CvRDT state $cState$ as well as totally-ordered updates $tUpdate$. We use the denotation $++$ to represent the combination of two sequences.

\section{Observable atomic consistency protocol (OACP)}\label{sec:oacp}

Following the definition of observable atomic consistency~\ref{oacpdefinition}, we now introduce a protocol that enforces this notion of consistency. We present the observable atomic consistency protocol (OACP) in four steps: first, we describe the so-called {\em data model}, similar to prior related work~\cite{BurckhardtLPF15}; second, we describe the client-side protocol; third, we describe the server-side protocol; finally, we discuss differences to the most closely related previous protocol.
{}
\paragraph{Data model.}
Operation types are classified according to different names. Thus, the function $\mathit{Eval}: Op^* \times Op \rightarrow \mathit{Val}$ presents the result when applying the sequence of operations on one object. We define abstract type $\mathit{CvUpdate}$ and $\mathit{TUpdate}$ stand for CvRDT updates and totally-ordered updates, $Value$ for the result from reading the log, $\mathit{Read}$ for an ordered read operation. Then we could give a data model using the function $\mathit{OACPVal}: \mathit{Read} \times \{\mathit{CvUpdate}, \mathit{TUpdate}\}^* \rightarrow \mathit{Value}$. The protocol will be explained in a client-server topology.

In the following, we use the type $\mathit{Log}$ to represent a log that records the (totally-ordered) sequence of all operations and states. Whenever a totally-ordered operation is committed, an entry is created and added to the log. Each log entry is of type $\mathit{Entry}$ which is defined as follows:

\begin{lstlisting}
class Entry { origin: Client, nextLogIndex: N, cState: CvRDT, op: TOp }
\end{lstlisting}
\noindent
Each entry contains a reference to the client that submitted the TOp (\verb|origin|); the index of the next log entry (\verb|nextLogIndex|); the state of the underlying CvRDT (\verb|cState|); and the committed totally-ordered operation (\verb|op|). We use the notation $++$ to represent the concatenation of two sequences.

\paragraph{OACP client protocol.}
 We now show the basic version of the OACP client-side protocol in Listing~\ref{list:egsp-client}. When CvOp invokes, the client will send the update to a random available server; while TOp happens, the client will submit the update to the current leader according to the chosen consensus protocol (in our case is Raft~\cite{OngaroO14}). The abstraction greatly simplifies the distinction between these two kinds of operations on the client side, so that it provides more flexibility for the extension on top of the client actor.

\begin{lstlisting}[numbers=left, numberstyle=\scriptsize\color{gray}\ttfamily, backgroundcolor = \color{white}, caption= Client-side OACP., label = list:egsp-client]
role OACP_Client {
  var result: Promise[Value];
  var response: Promise[Value];

  // client interface
  CvOp(u: CvUpdate) {
    response := new Promise[Value];
    CRDT_submit(u); // send to random replica
    return response;
  }

  TOp(msg: TUpdate | Read): Promise[Value] {
    result := new Promise[Value];
    RTOB_submit(msg); // send to leader replica
    return result;
  }

  // network interface
  onReceive(log: Log) {
    result.complete(OACPVal(updates(log)));
  }

  onReceive(response) {
    response.complete()
  }

  function updates(l: Log): (TUpdate*, cState) = {
    return (l[0].tUpdate ++ ... ++ l[l.length-1].tUpdate, l.cState);
  }
}
\end{lstlisting}
\noindent
An important assumption of the protocol is that the client never invokes an operation before the promise of a previous invocation has been completed, especially when a TOp happens directly after a CvOp, the client needs an acknowledge of the previous message from the server side to invoke the TOp. In this case, the program order on the client side will also be preserved on the server side.

\paragraph{OACP server protocol.}
We now move on to the server side of the protocol. We assume that the application has continuous access to the network. This mirrors the practical usage of many existing applications, such as chat and Twitter-like micro-blogging.

In Listing~\ref{list:egsp-server}, $\mathit{CRDT.merge}$ is an abstract merge operation for any type of CRDT; $\mathit{Log.nextIndex}$ is the next store index for $\mathit{Entry}$ according to Raft protocol; $\mathit{cStateCollect()}$ is a message for getting all the $cState$ from other servers; $\mathit{Broadcast()}$ is a message which will be sent to all the other servers in the cluster; $\mathit{RTOB\_submit}$ wraps an $\mathit{Entry}$ structure which contains the client, the number of current log index, the current monotonic state and the update from the client, add the $\mathit{Entry}$ in the log and broadcast to all the other servers to keep consensus on each replica.

When the server receives the CvOp, it will merge the current CvRDT state and broadcast the change to all the other servers; when TOp is received, the leader server will collect the current states from all the replicas and make an RTOB so that each replica keeps the same log. RTOB in OACP system is implement by Raft consensus protocol, and each server will do commit when they get confirmation from the leader server. When $\mathit{TUpdate}$ and $\mathit{Read}$ handled by non-leader server, it will forward the update to the current leader in the cluster. Moreover we define $\mathit{Read}$ as TOp to guarantee for getting the newest result directly from the leader.

In the protocol description, we set up a special flag ``frozen''. When a TOp is taking place, the flag will turn to true for each replica. For coming operations after that, the system will stash all the updates without changing the current state to keep OAC for all replicas.

\begin{lstlisting}[numbers=left, numberstyle=\scriptsize\color{gray}\ttfamily, backgroundcolor = \color{white}, caption = Server-side OACP (unoptimized)., label = list:egsp-server]
role OACP_Server {
  var currentState: CRDT;
  var log: Log;
  var currentLeader: Server;
  var frozen: Boolean;
  var result: Promise[Value];

  onReceive(u: CvUpdate) {
    if frozen then { buffer.stash(u); }
    else {
      currentState = CRDT.merge(currentState, u);
      Broadcast(currentState);
      client.reply(); //acknowledge to client
    }
  }

  onReceive(msg: TUpdate | Read) {
    if currentRol.isLeader && frozen then { stash(msg); }
    else if (currentRole.isLeader) {
      frozen = true;
      numStateMsgReceived = 0;
      result = new Promise[Value];
      result.onSuccess { v => client.reply(v);}
      Broadcast(GetState);
    }
    else {forward(currentLeader, msg);}
  }

  onReceive(msg: GetState) {
    frozen = true;
    reply(StateIs(currentState));
  }

  onReceive(msg: StateIs) {
    if currentRole.isLeader then {
      numStateMsgReceived += 1;
      currentState = CRDT.merge(currentState, msg.cState);
      if numStateMsgReceived == numReplicas-1 then {
        RTOB(new Entry {
            origin = msg.sender,
            number = Log.nextIndex(log),
            cState = currentState,
            toUpdate = msg });
        Broadcast(Melt);
        result.complete(log);
      }
      if timeout then { // fault handler
        RTOB(new Entry {
            origin = msg.sender,
            number = Log.nextIndex(log),
            cState = Log.cState,
            toUpdate = Recovery });
        Broadcast(Melt);
        result.complete(failure);
        }
      }
    }
  }

  onReceive(msg: Melt) {
    frozen = false;
    buffer.unstash() or discard();
  }
}
\end{lstlisting}

When leader receives acknowledges from majority number of servers, a ``Melt'' message will be sent around to set the ``frozen'' flag to false again. If there are $n$ servers in the cluster, then there will be $2(n-1)$ messages adding to the whole protocol. Thus we come up with the optimized version (see List~\ref{list:egsp-server-op}) to turn the flag to false each time when a server makes a commit in the consensus protocol.

\begin{lstlisting}[numbers=left, numberstyle=\scriptsize\color{gray}\ttfamily, backgroundcolor = \color{white}, caption = Server-side OACP (optimized), label = list:egsp-server-op]
  onCommit(e: Entry) {
    frozen = false;
    buffer.unstash() or discard();
  }
\end{lstlisting}

\paragraph{Optimisation considering the sequence of operations}\label{section:optimization}
Then we go deeper to the inside of OACP. It’s the combination of two types of systems, one is the representative of on-demand strong consistency and one is on behalf of a popular way to achieve strong eventual consistency. Taken the low latency from CRDT together with the consensus from Raft, OACP takes the advantage of both systems. When the two systems need to interface, we need to think about how the system behave.

Now let’s discuss about the how the system behaviour change when two different types of operations happen in certain sequence. There’re four situations which shows in Table~\ref{table:sequence}.
\begin{table}[ht]
\centering
\begin{tabular}{c l l }
  \hline
  No. & Sequence & System Behaviour \\
  \hline
  (1) & CvOp $\rightarrow$ CvOp & CvRDT operation (add, remove, merge) \\
  (2) & TOp $\rightarrow$ TOp & OACP protocol optimization (discussed in Section~\ref{section:twittercompare}) \\
  (3) & CvOp $\rightarrow$ TOp & Freeze system $\rightarrow$ Total order broadcast (implemented by Raft) \\
  (4) & TOp $\rightarrow$ CvOp & Melt system $\rightarrow$ CvRDT operation \\
  \hline
\end{tabular}
\caption{System behaviour in different sequence of operations}\label{table:sequence}
\end{table}

We could illustrate the above chart as follows: (1) CvOp will use CvRDT property to operate so that the order of the operation can be disordered, the system is behaving like the CvRDT system. (2) When two TOps occurs, then the system will behave according to raft protocol. (3) When a TOp happens after the CvOp, first the system need to stop receiving any coming CvOps ("Frozen" state) and then using elected leader to decide the total order for the current state. (4) When a CvOp wants to take effect after the TOp, then it should make the system "Melt" first, and then behave like what described in (1). The interesting aspect we observe when we make this classification is that in situation (3) and (4), we need to deal with the shift state ("Frozen" $\leftrightarrow$ "Melt") of the system. It can be optimised by introducing another parameter called "auto-melt". The property of "auto-melt" is that once defined, the system will decide whether to melt the system automatically. Now le'ts have a look at two examples about how the argument works in different selections.

In the shopping cart example, when we checkout at certain point, the user could still be able to do add or remove operations to another shopping cart. These following MonOps will not be considered for this specific "checkout", but after the checkout, the updates should be able to take effect on the system. At this point, we need to use auto-melt so that the CvOps can take effect after certain TOps. In another example such as a simple bidding system, we could define all the "Bid" operations as CvOps and the final "WinnerCheck" is the TOp. When the "WinnerCheck" is called, the system will freeze as usual but the following "Bid" will be simply dropped. Another bid need to be started manually by calling a special "Melt" message to the system.
\begin{figure}[ht]
\centering
\includegraphics[width=5.00in]{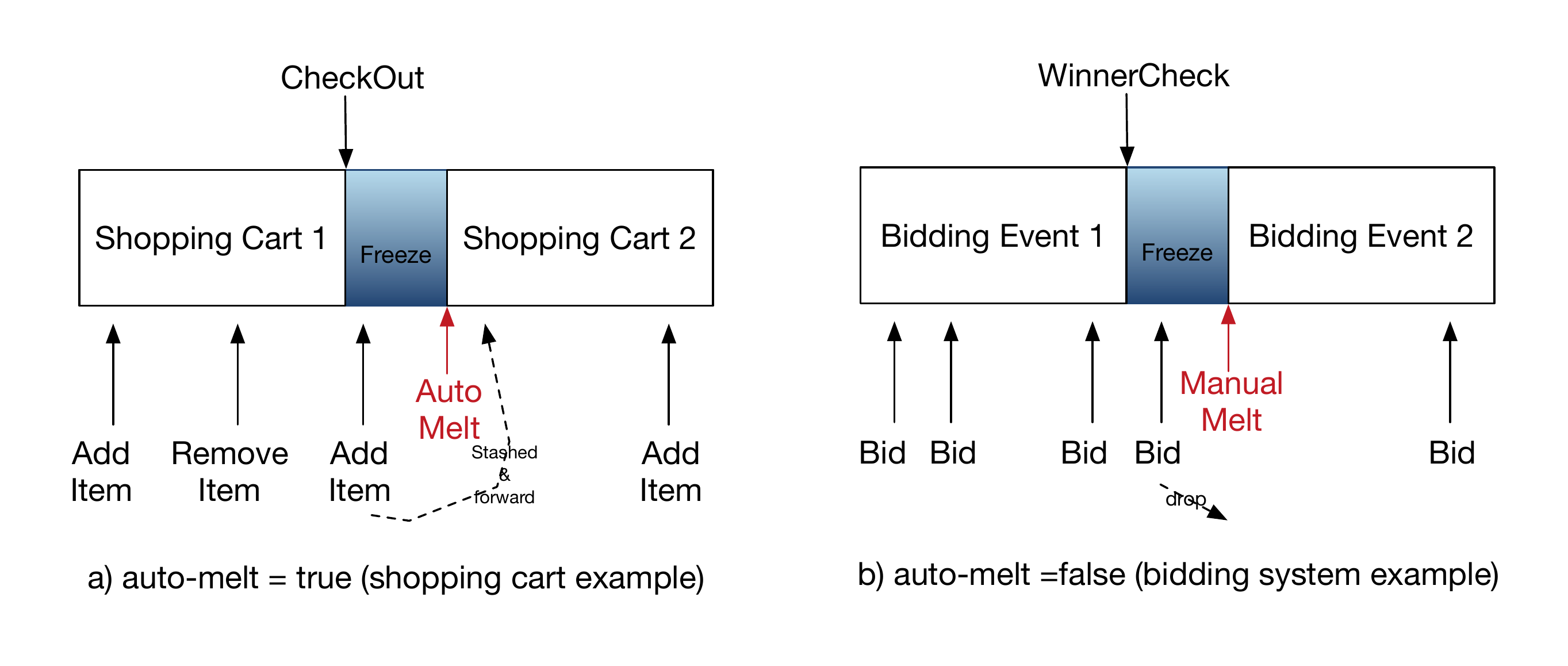}
\caption{Automelt example: Shopping cart vs. Bidding system}\label{graph:automelt}
\end{figure}

\paragraph*{Fault Model.} During the process of RTOB, the failure of the system will be handled by Raft protocol where $n$ nodes crash failure will be tolerated in a $2n+1$ nodes cluster. During the state gathering phase (See Line 38, 47), if the receiving time of leader exceeds the setup timeout duration, then failure recovery strategy is to retrieve the cState in the last log entry and synchronize all the replicas. The OACP protocol provides a neat solution to avoid staleness read in CvRDTs.

\paragraph*{Comparison to GSP.}
GSP is an operational model for replicated shared data. Figure~\ref{graph:GSP} shows the abstraction of GSP protocol. It supports update and read operations from the client. Update operations are stored both in the local pending buffer so that when a read happens, it will perform ``read your own writes'' directly from the local storage. That property makes offline read possible so that even when the network is broken, the application can work properly and the following updates will be stashed in the buffer and resent until the network is recovered again.

\begin{figure}[!ht]
\centering
\includegraphics[width=2.50in]{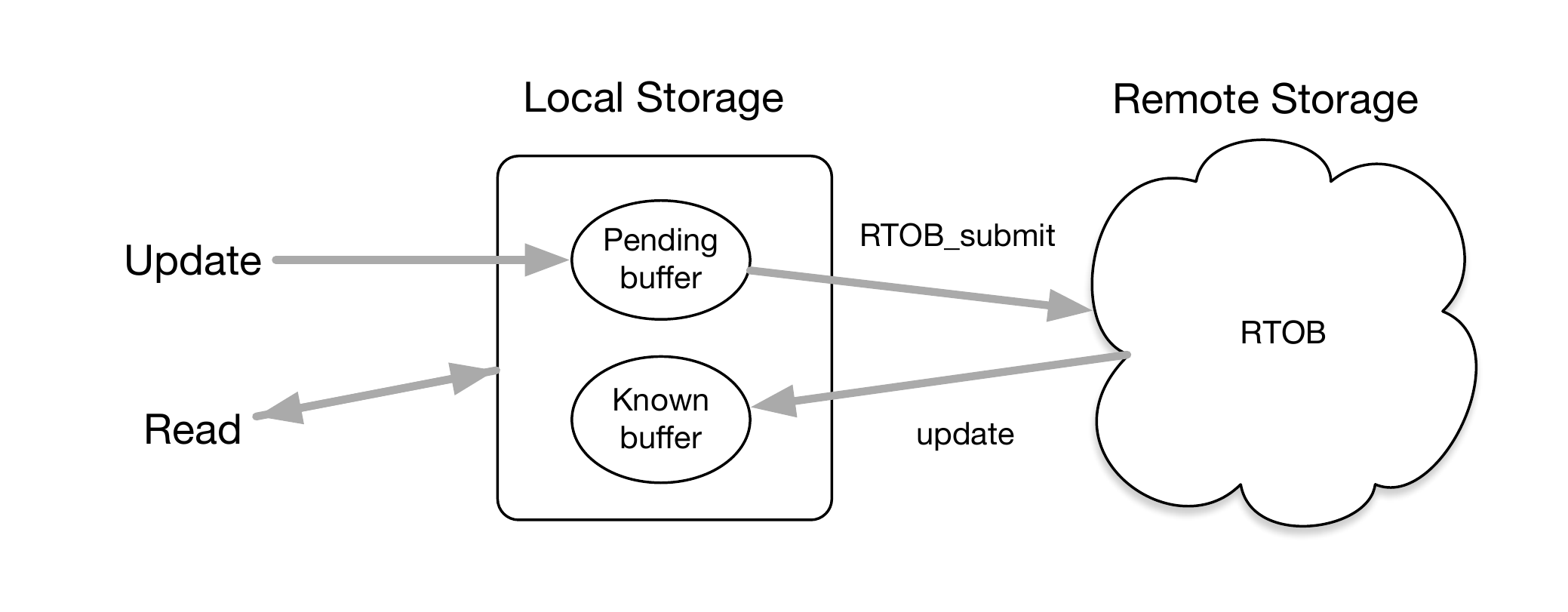}
\caption{GSP protocol abstraction}\label{graph:GSP}
\end{figure}

\begin{figure}[!ht]
\begin{lstlisting}[numbers=left, numberstyle=\scriptsize\color{gray}\ttfamily, backgroundcolor = \color{white}, escapeinside={(*}{*)}]
role Core_GSP_Client{
  known: Round * := [];
  pending: Round * := [];
  round: N := 0;

  //client interface
  update(u: Update) {
    (*\highlight{\texttt{pending := pending . u;}}*)
    RTOB_submit(new Round{origin = this, number = round ++, update = u});
  }

  read(r: Read): Value {
    var compositelog := known . pending;
    return GSPValue(r, updates(compositelog));
  }

  //network interface
  onReceive(r: Round) {
    known := known . r
    if (r.origin = this) {
      assert(r = pending[0]);
      pending := pending[1...];
    }
  }

  class Round {origin: Client, number: N, update: Update}
  function updates(s: Round*): Update* {return s[0].update ... s[s.length - 1].update;}
}
\end{lstlisting}
\caption{Core Global Sequence Protocol (GSP)~\cite{BurckhardtLPF15}}\label{list:gsp}
\end{figure}

GSP also relies on RTOB to send back a totally ordered sequence of updates to the locally known buffer. The existence of known and pending buffer provides the possibility for updating these buffers fully asynchronously. In order to achieve high throughput, it also provides the batching option so that it does not require RTOB for every operation.

Compared with OACP, GSP not only has only totally-ordered operations but also performs locally read operation. When the relative order of operations of different clients is observable, it will provide some difficult condition. Given the example in paper~\cite{BurckhardtLPF15}:

\begin{tabular}{ l l }
  $wr(A, 2)$ & . \\
  . & $wr(B, 1)$ \\
  . & $wr(A, 1)$ \\
  . & $rd(A) \rightarrow 2$ \\
  $rd(B) \rightarrow 0$ & . \\
\end{tabular}

Above is the key-value store shared data model which initially stores 0 for each address. In GSP protocol, such interleaving will be possible since $rd(B)$ can get 0 before the local storage get the update from $wr(B, 1)$.

Local read from GSP has a strange effect that it can not provide the strong consistency when processing locally read. While in OACP, this confusion will not happen since we always process read as a TOp and it will always get the updated state from the server. So only the following is possible, and the observable strong consistency is always preserved:

\begin{tabular}{ l l }
  $wr(A, 2)$ & . \\
  . & $wr(B, 1)$ \\
  . & $wr(A, 1)$ \\
  . & $rd(A) \rightarrow 1$ \\
  $rd(B) \rightarrow 1$ & . \\
\end{tabular}

\section{Performance evaluation}\label{sec:eval}

We evaluate the performance of OACP from the perspective of latency, throughput and coordination. We use a microbenchmark as well as a Twitter-like application inspired by Twissandra~\cite{Twissandra}.

\paragraph{Experimental set-up.}
Our OACP implementation is based on the cluster extension of the widely-used Akka~\cite{Akka} actor-based middleware. The cluster environment is configured using three seed nodes; each seed node runs an actor that detects changes in cluster membership (i.e., nodes joining or leaving the cluster). This enables freely adding and removing cluster nodes.


The experiments on coordination are performed on a 1.6 GHz Intel Core i5 processor with 8 GB 1600 MHz DDR3 memory running macOS 10.13. The experiments on latency and throughput are performed on Amazon EC2 cluster which includes three T2 micro instances (1 vCPU, 1 Gb memory, and EBS-only storage. ) running in Ohio, London, and Syndey. Figure~\ref{table:roundtrip} shows the average round-trip latency between each pair of sites. Our implementation is based on Scala 2.11.8, Akka 2.4.12, AspectJ 1.8.10, and JDK 1.8.0\_162 (build 25.162-b01). The OACP implementation is available open-source on GitHub.\footnote{See \url{https://github.com/CynthiaZ92/OACP}}

\begin{figure}
\begin{center}
\begin{tabular}{ |c|c|c|c| }
 \hline
  & Ohio & London & Sydney \\
  \hline
 Ohio & 0.53ms & 85.6ms & 194ms \\
 \hline
 London &  & 0.42ms & 279ms \\
 \hline
 Sydney &  &  & 0.88ms \\
 \hline
\end{tabular}
\end{center}
\caption{Average round-trip latency between Amazon sites} \label{table:roundtrip}
\end{figure}

\begin{figure}
\centering
\begin{subfigure}[b]{0.4\textwidth}
  \includegraphics[width=\textwidth]{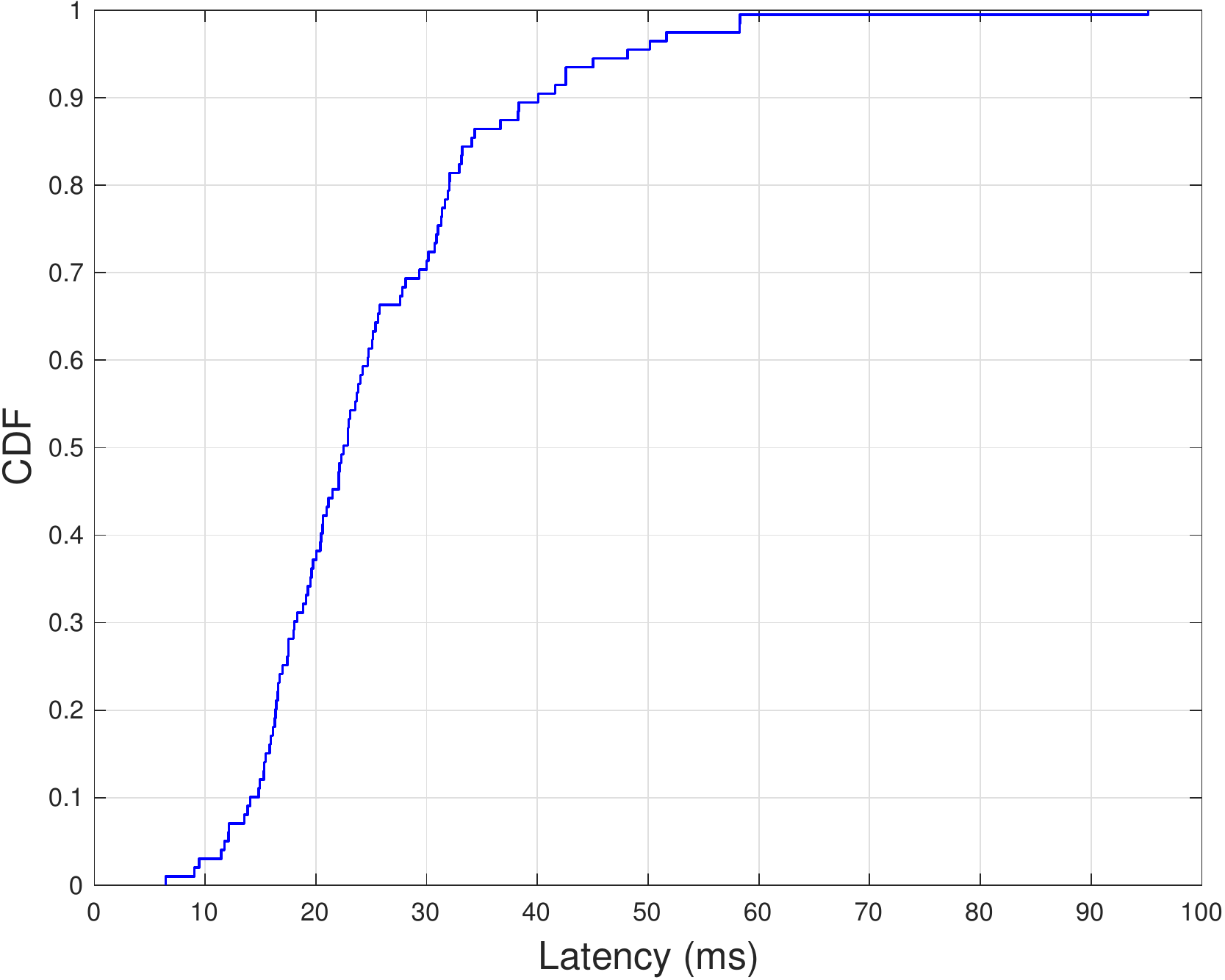}
  \caption{CvOp latency}\label{graph:CvOplatency}
  \end{subfigure}
\qquad
  \begin{subfigure}[b]{0.4\textwidth}
  \includegraphics[width=\textwidth]{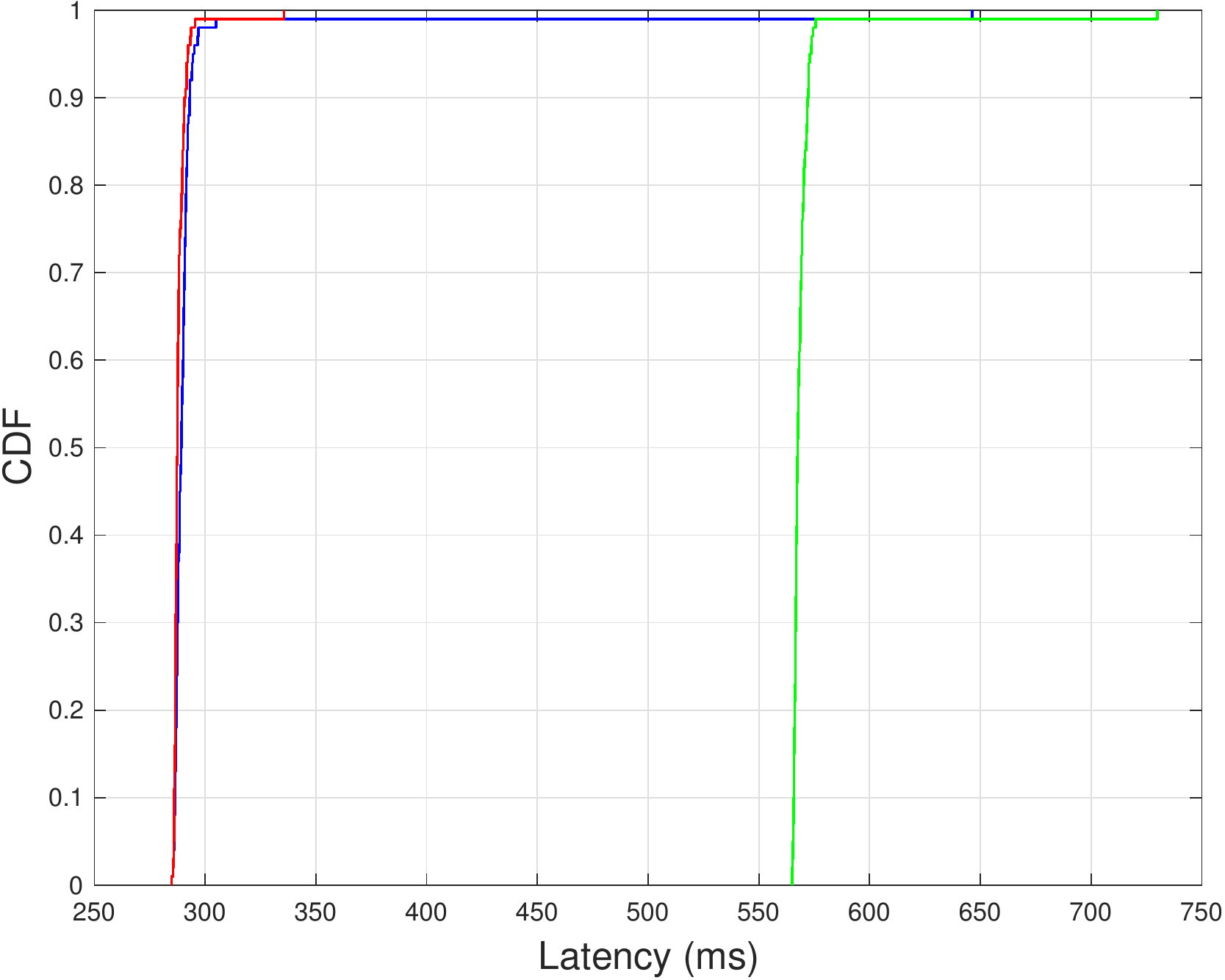}
\caption{TOp latency}\label{graph:TOplatency}
  \end{subfigure}
\caption{Latency CDF for CvOps and TOps when leader node locates on different regions. In (b), the green line cooresponds to the condition where leader locates in Sydney, the blue line cooresponds to Ohio and the red line cooresponds to London. }\label{graph:cdf}
\end{figure}

\paragraph{Latency.}
In OACP, any CvOp gets an immediate response once the request arrives to any of the servers in the cluster. In contrast, the TOp runs consensus protocol underneath to achieve consistency. In order to understand the effect of CvOps and TOps, we measure the latency for CvOps and TOps on Amazon EC2 in three different regions: Ohio, London and Sydney and we plot the CDFs of observed latencies in Figure~\ref{graph:cdf}. The cluster in our experiments is consists of three nodes which locate in different regions. A leader node needs to be elected to keep consensus. We put the client node in London and measure different conditions when the leader node locates in different regions. In general, CvOps get quick responses as we can see that the maximum latency is 60 ms, and 90\% of the response latency are within 40 ms. The latency time of TOps depends on the location of the leader node.  When the leader node locates in Sydney, the maximum latency is 600 ms. And when the leader node locates in the other two regions, the maximum latency is around 350ms.

\paragraph{Throughput.}
Now we focus on the throughput of OACP. We generate benchmarks with different proportion of CvOps and TOps. While we increase the number of concurrent requests to the same consistent log in the cluster, we measure the duration for processing all of the requests. The throughput is then the number of requests divided by the duration. The results in Figure~\ref{graph:throughput} show that increasing the ratio of TOps decreases the throughput. TOps require the leader node to force all the other replicas to reach consensus on the same log. Thus there will be a request queue on the server side. Any of the server nodes can process CvOps, and the commutative property of CvOps allows them to be processed in a random order. The throughput of the mix workloads is located between the pure workloads which give the programmer a range of choices.

\paragraph{Coordination.}
We consider this aspect because previous work has shown that reducing the coordination within the protocol can improve the throughput of user operations dramatically~\cite{BernsteinBBCFKK17}. In order to evaluate the performance independent of specific hardware and cluster configurations, our experiments count {\em the number of exchanged messages}. The message counting logic is added via automatic instrumentation of the executed JVM bytecode using AspectJ~\cite{KiczalesHHKPG01}.

\begin{figure}
    \centering
    \includegraphics[width=3.00in]{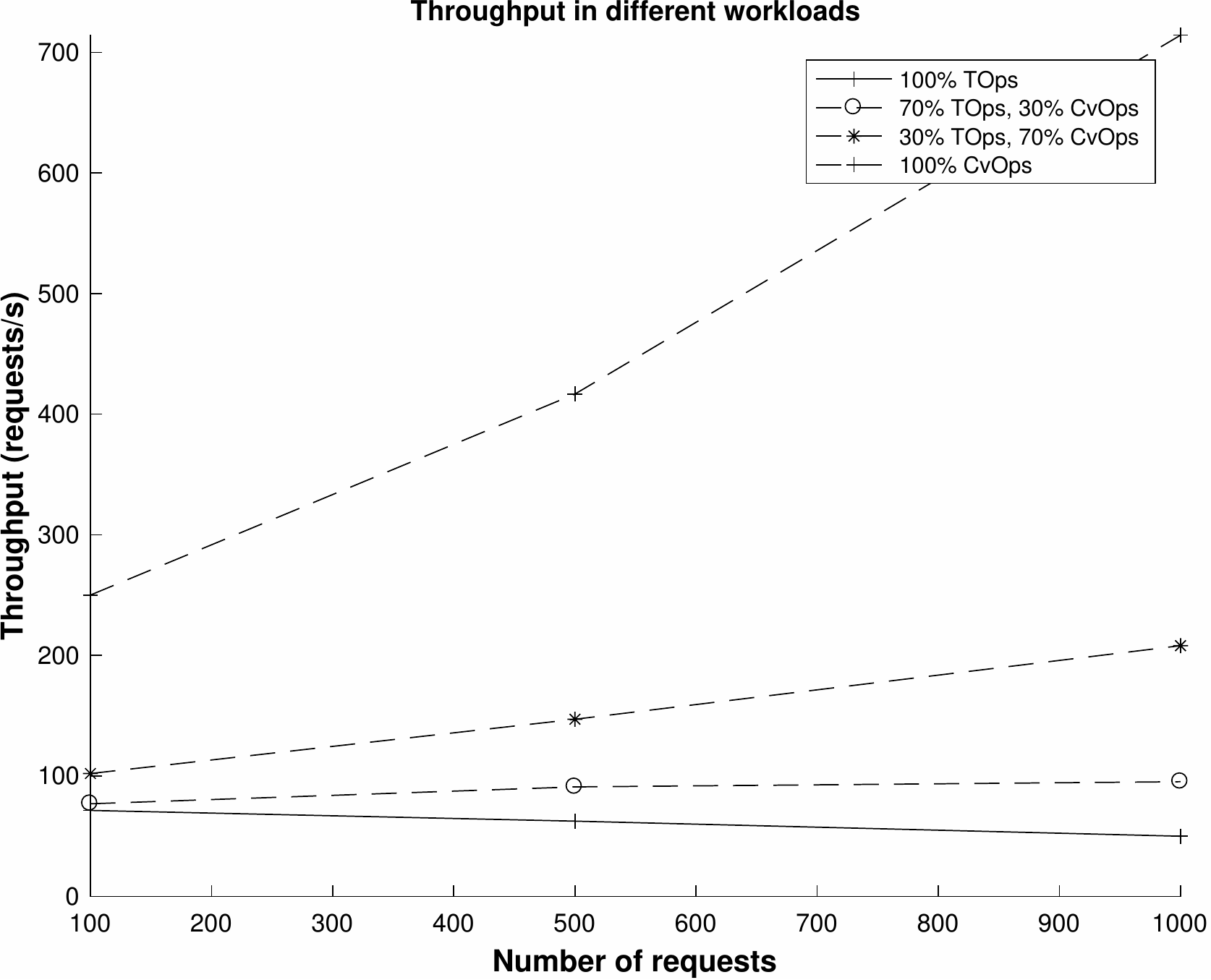}
    \caption{Throughput for a 3 site cluster with varying CvOp and TOp workload mixes.}\label{graph:throughput}
\end{figure}

\subsubsection{Microbenchmark.}
We start the evaluation with a simple shopping cart benchmark to see the advantages and weakness of the OACP protocol. We define the ``add'' and the ``remove'' operation as CvOps in OACP. The ``checkout'' operation in OACP is defined as a TOp (to ensure consistency upon checkout). Then we generate sequences of $n$ operations where $n \in (0, 1000]$; each operation can be either ``add'', ``remove'', or ``checkout''.

We use Akka's messaging interface to define an AspectJ pointcut (Figure~\ref{fig:aspectj}, lines 1--2). This allows us to detect every message sent between actors (Figure~\ref{fig:aspectj}, lines 4--5). In order to count only messages generated by each protocol, we filter out internal Akka messages (Figure~\ref{fig:aspectj}, line 9). Furthermore, we filter out heartbeat messages, since all protocols are assumed to exchange the same number of heartbeat messages, e.g., to implement RTOB.

\begin{figure}[h!]
\begin{lstlisting}[numbers=left, numberstyle=\scriptsize\color{gray}\ttfamily, backgroundcolor = \color{white}]
@Pointcut(value = "execution (* akka.actor.ActorCell.receiveMessage(..))")
    public void receiveMessagePointcut(Object msg) {}

@Before(value = "receiveMessagePointcut(msg)")
    messages.counting(msg);

class ActorSystemMessages {
  void counting(Object msg) {
    filter(akka internal msg && heartbeat msg)
}
\end{lstlisting}
\caption{AspectJ pointcut pseudocode.}\label{fig:aspectj}
\end{figure}

We compare the performance of OACP with the other two baseline protocols which are described as follows.

\begin{itemize}
\item Baseline protocol: every operation is submitted using RTOB to keep consistency on all the replicas.
\item Baseline protocol with batching: an optimized version which gives a fixed batching buffer and allows to submit multiple buffered operations at once.
\item OACP protocol: the protocol as described in Section~\ref{sec:oacp}.
\end{itemize}

We compared the number of exchanged messages among the baseline protocol, the batching protocol (batching buffer size: 5000 operations), and the OACP protocol, using a 3-node cluster. The results are shown in Figure~\ref{graph:GSPOACP}.

\begin{figure}
\centering
\begin{subfigure}[b]{0.46\textwidth}
  \includegraphics[width=\textwidth]{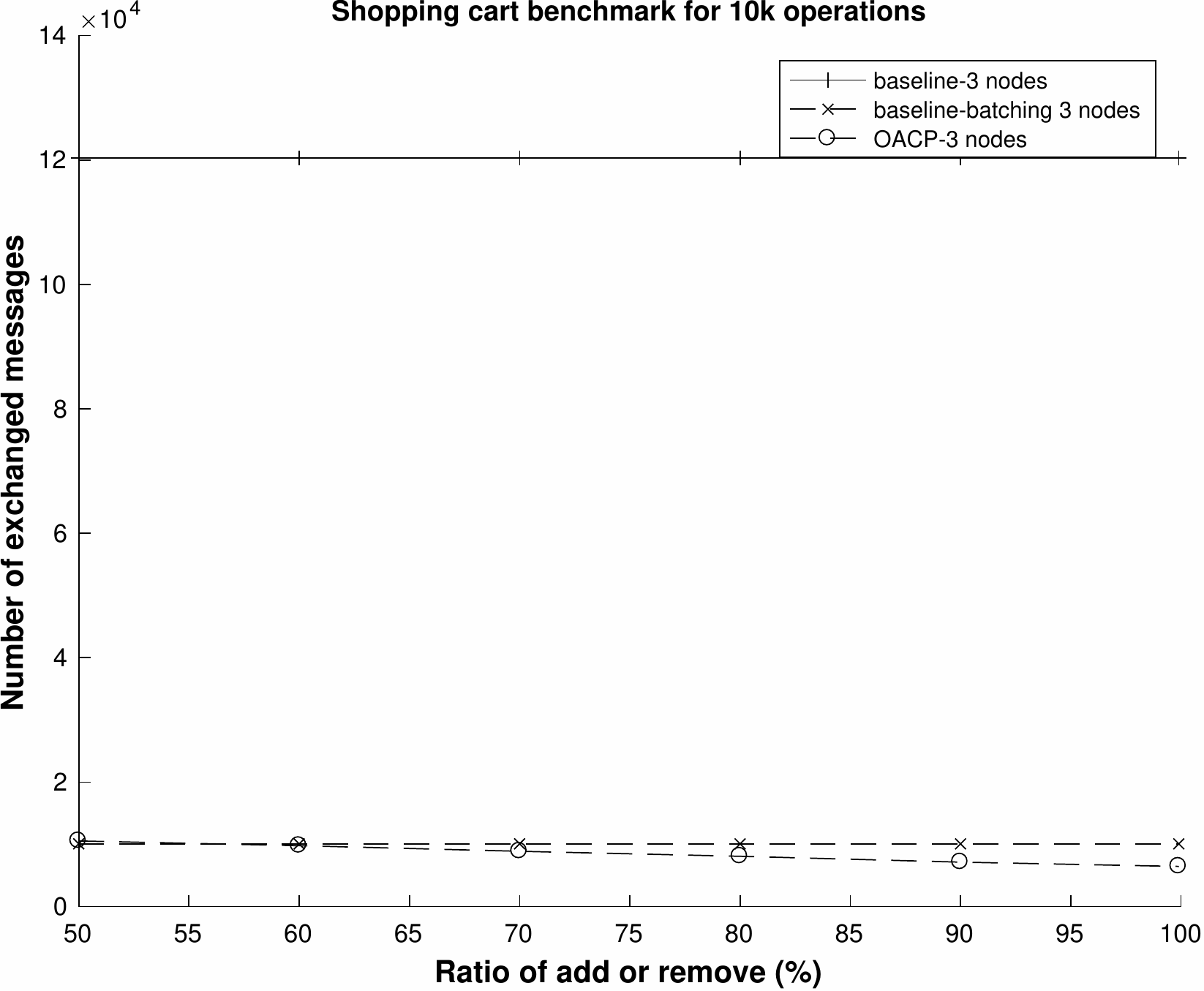}
  \caption{}\label{graph:GSPOACP}
  \end{subfigure}
\qquad
  \begin{subfigure}[b]{0.46\textwidth}
  \includegraphics[width=\textwidth]{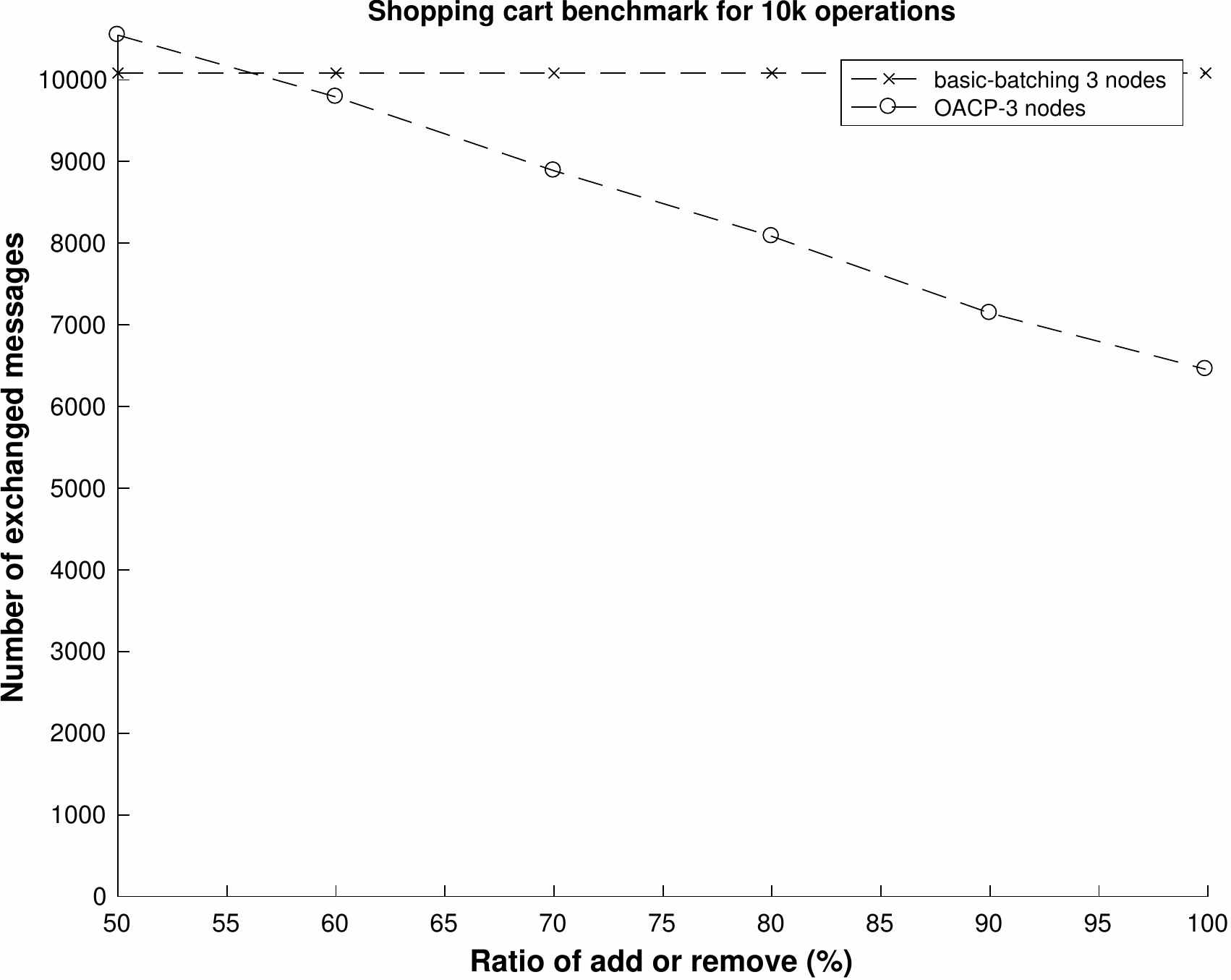}
\caption{}\label{graph:BatchingOACP}
  \end{subfigure}
\caption{Comparison on coordination among different protocols}
\end{figure}

The x-axis represents the ratio of ``add'' or ``remove'' operations in the whole sequence of operations (10k requests in this case), the y-axis represents the number of exchanged messages. From Figure~\ref{graph:GSPOACP}, we can see that baseline protocol requires a much higher number of exchanged messages, namely 12$\times$ the number of messages compared to both the batching protocol and the OACP protocol. In the more detailed Figure~\ref{graph:BatchingOACP}, we can see that when the ratio of CvOps increases, OACP benefits more. In OACP, when CvOp is 90\%, the number of messages can be reduced for 30\% compared with the batching protocol.

These results suggest the following {\em guidelines} for helping developers choose which protocol to use. When the percentage of CvOps is low (less than 50\% in the above microbenchmark), i.e., when the application needs TOps quite frequently, then batching for TOps is a better choice. When there are more CvOps happening between two TOps, then OACP performs better.

\paragraph{Scalability}

We also made a test from the scalability perspective. We increase the number of nodes in baseline protocol and OACP from 3 to 7, and at the meantime increase the number of operations,we get the Figure~\ref{graph:gspvsegsp}. When we abstract the gradian from Figure~\ref{fig:scalability}, we could get a scalability trend for these two systems.

\begin{figure}
    \centering
    \begin{subfigure}[b]{0.4\textwidth}
        \includegraphics[width=\textwidth]{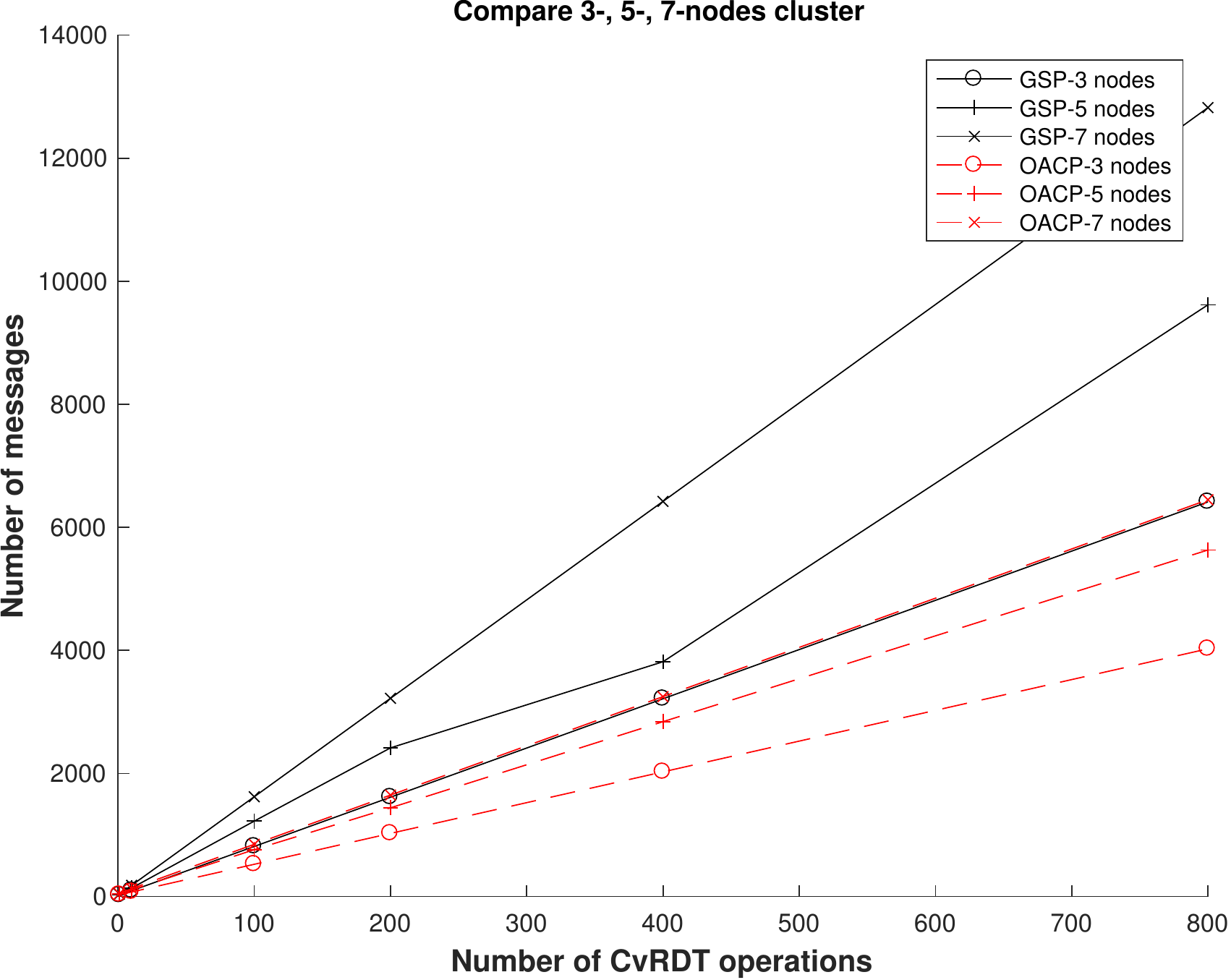}
        \caption{Compare 3-, 5-, 7-nodes cluster}
        \label{fig:compare}
    \end{subfigure}%
    \qquad 
    \begin{subfigure}[b]{0.4\textwidth}
        \includegraphics[width=\textwidth]{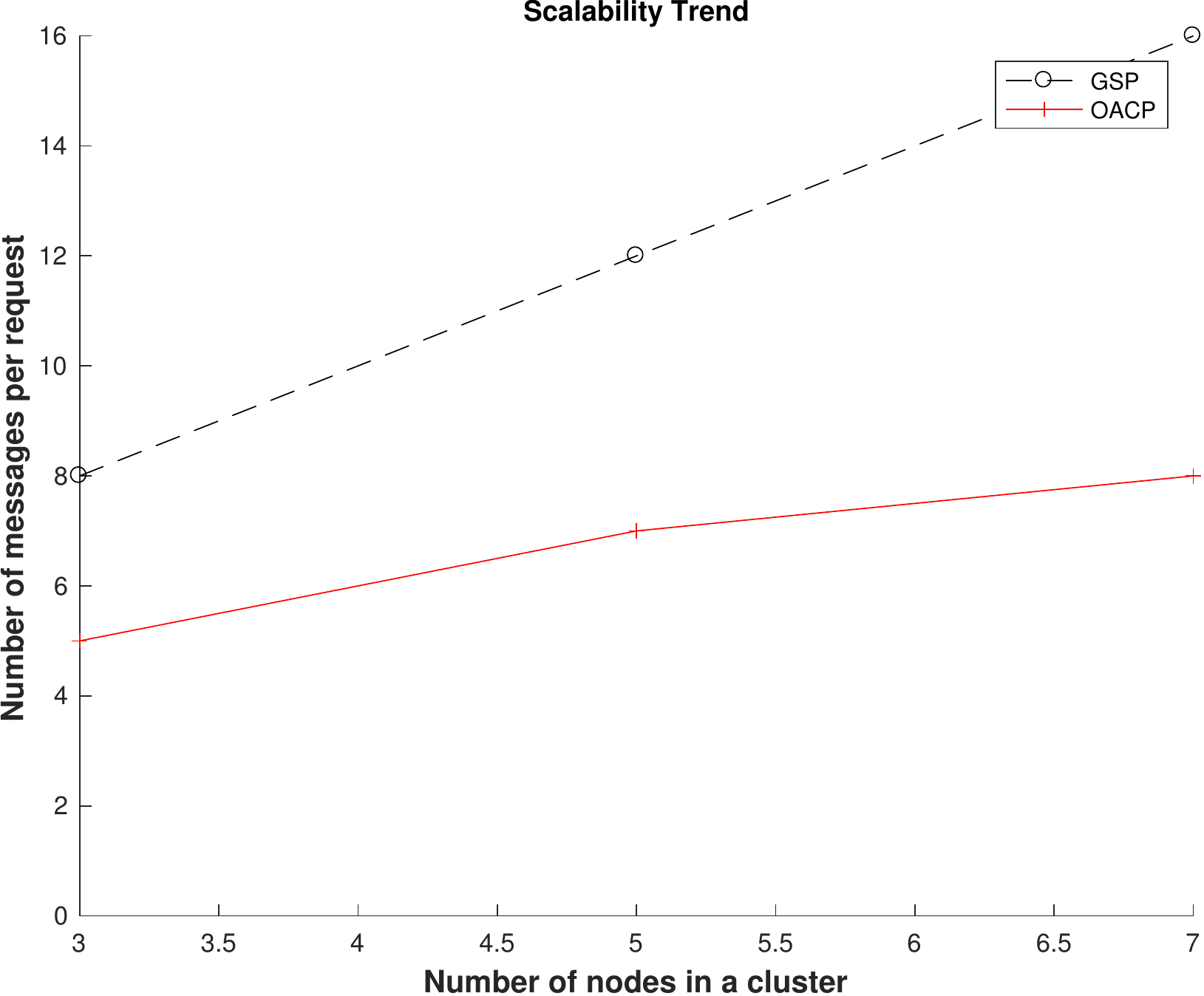}
        \caption{Trend of scalability}
        \label{fig:scalability}
    \end{subfigure}
    \caption{Baseline protocol vs. OACP on scalability}\label{graph:gspvsegsp}
\end{figure}

As we can see, the scalability of OACP is also better than GSP since it won’t be effected too much by the number of nodes in a cluster.

\subsubsection{Case study: Twitter-like application}
Following the shopping cart microbenchmark, we now extend our experiments to a more realistic application. This also allows us to investigate more aspects of our system since the application makes use of all features of OACP. We define a simple Twitter-like social networking application which supports \textsc{AddFollower}, \textsc{Tweet}, and \textsc{Read} operations. We define \textsc{Tweet} and \textsc{Read} as TOps and \textsc{AddFollower} as a CvOp in OACP. In the benchmarks, we focus on one specific user with a certain number of followers and send tweets continuously when the state of followers does not change.

\paragraph{Optimized observable atomic consistency protocol: O$^2$ACP.}\label{section:twittercompare}
For Twitter-like applications, the frequency of different events varies for different users. Some popular accounts tend to tweet more and to follow fewer users, while some newcomers follow more and tweet less.

Consider the following operation sequence:\newline
(1) \textsc{Tweet} $\rightarrow$ (2) \textsc{Tweet} $\rightarrow$ (3) \textsc{AddFollower} $\rightarrow$ (4) \textsc{Tweet}\newline
\noindent
Since we define \textsc{Tweet} as a TOp, the state of all replicas is consistent after each \textsc{Tweet} operation. Therefore, there is no need for a merge operation to be executed between (1) and (2) (to merge the convergent states of replicas). However, in the case (3) $\rightarrow$ (4), the state updated by (3) first needs to be merged into all replicas before executing (4). Thus, one possible optimization is to notify the system of the sequence of operations, so that when two TOps happen consecutively (e.g., (1) $\rightarrow$ (2)), the OACP system does not need to gather state information first. In this way, the system can decrease the number of exchanged messages. We call this optimization O$^2$ACP; the results of this comparison are shown in Figure~\ref{fig:e2gsp}.

\begin{figure}[t]
    \centering
    \begin{subfigure}[b]{0.4\textwidth}
        \includegraphics[width=\textwidth]{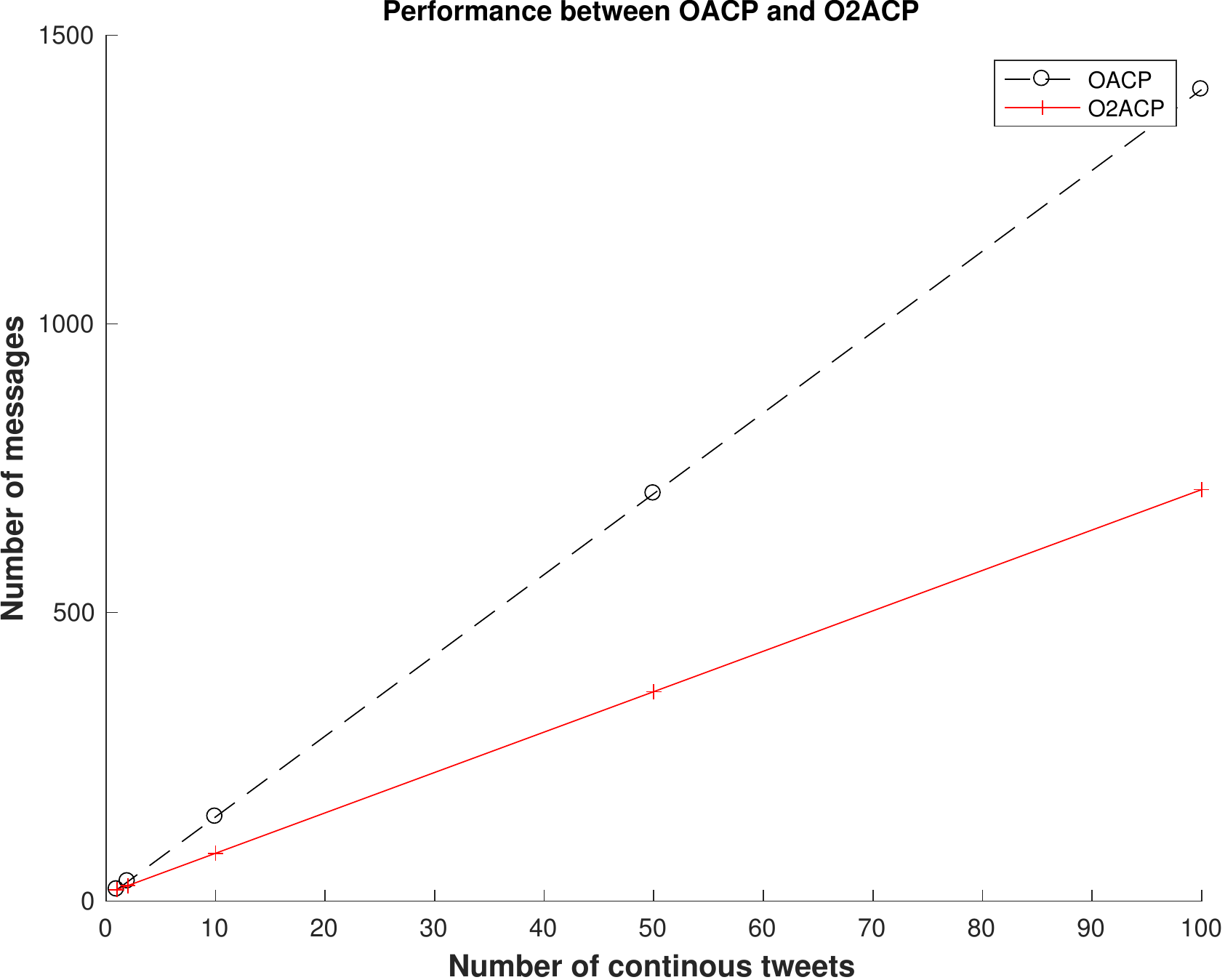}
        \caption{}
        \label{fig:e2gsp}
    \end{subfigure}%
    \qquad 
    \begin{subfigure}[b]{0.4\textwidth}
        \includegraphics[width=\textwidth]{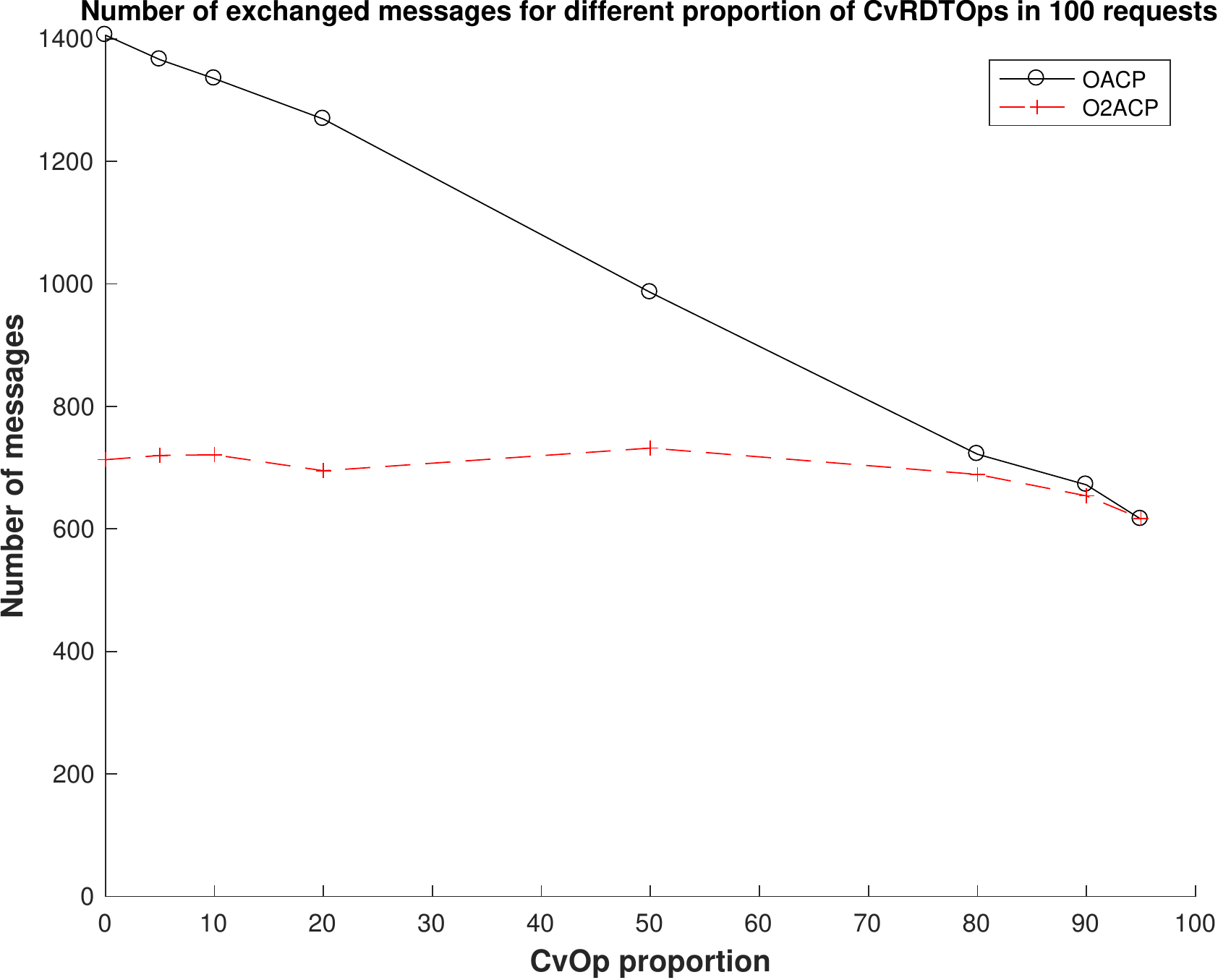}
        \caption{}
        \label{fig:proportion}
    \end{subfigure}
    \caption{OACP vs. O$^2$ACP}\label{graph:egspg2sp}
\end{figure}

According to Figure~\ref{fig:e2gsp}, when the number of continuous tweets grows, the increase of messages is significantly smaller in the case of O$^2$ACP than in the case of OACP. In the case of 100 tweets, O$^2$ACP only requires about 50\% of the messages of OACP, which is a significant improvement. We made another measurement for random sequences of TOps and CvOps, by varying the proportion of CvOps between 0\% and 95\%. The total number of client requests is 100. The results are shown in Figure~\ref{fig:proportion}. In each case, we take the average of 10 measurements for each proportion. The proportion of CvOps has a more substantial effect in the case of OACP than in the case of O$^2$ACP, since when CvOps increase from 0\% to 95\%, the exchanged messages in OACP reduce from around 1400 to 600, while in O$^2$ACP, the number of exchanged messages remains quite stable. This means that the performance of the optimized system is more stable in various situations.

\section{Related work}\label{sec:rel-work}

\paragraph*{Consistency levels.}
CAP~\cite{GilbertL02} theorem points out the impossibility for any distributed system to achieve consistency, availability and partition tolerance at the same time. Zookeeper~\cite{HuntKJR10} provides sequential consistency~\cite{goodman89} that updates from a client will be applied in the order that they were sent. Bayou~\cite{TerDB} is designed for supporting real-time collaborative applications and thus gives up consistency for high availability, providing eventual consistency~\cite{Burckhardt13}. Similarly, several other systems also provide multiple levels of consistency in order to give more flexibility according to application requirements and the network environment. \cite{Kraska09} allows a user to define the consistency guarantees as well as switch the guarantees at runtime automatically. Simba~\cite{Perkins15} enables mobile apps to select data consistency levels enforced by the system.
There are also many classifications in multi-level consistency. Fork consistency~\cite{Li04a,MazSha02} allows users to read from {\em forked sites} that may not be up-to-date. In contrast, write operations require all sites to be updated. Red-Blue consistency~\cite{LiPCGPR12} provides two types of operations: red operations and blue operations. SIEVE~\cite{LCPRV14} is a system based on RedBlue consistency which automatically chooses consistency levels according to a user's definition of system invariants. Explicit consistency in Indigo~\cite{Balegas15} guarantees the preservation of specific invariants to strengthen consistency beyond eventual consistency.

\paragraph*{Distributed application development frameworks.}
Correctables~\cite{Guerraoui16} is an abstraction to decouple applications from their underlying database, also provides incremental consistency guarantees to compose multiple consistency levels. QUELEA~\cite{Sivaramakrishnan15} gives a more well-reasoned specification for achieving the right consistency level. GSP~\cite{BurckhardtLPF15,MelgrattiR16} provides an operational reference model for replicated shared data. It abstracts the data model so that it can be applied to different kinds of data structures. In our companion technical report~\cite{ZhaoH18}, we compare the consistency model of GSP with OACP. Orleans~\cite{orleans} abstracts virtual actors to model distributed systems problems. The Akka framework~\cite{Akka,Haller12} provides a widely-used implementation of the actor model~\cite{agha86} on the JVM for writing highly concurrent, distributed, and resilient applications.

\paragraph*{Data management in distributed systems.}
The consensus of data in each replica is achieved by using consensus algorithms. Paxos~\cite{Lamport05} and Raft~\cite{OngaroO14} are popular protocols to manage consensus in replicated systems, used to implement RTOB~\cite{DefSchUrb04}. The ZooKeeper Atomic Broadcast protocol~\cite{JunqueiraRS11} (Zab) guarantees the replication order in ZooKeeper using Paxos. Our RTOB implementation is inspired by Zab but uses Raft because of its simplicity.
For resolving shared data conflicts in distributed systems, there are several approaches. CRDTs~\cite{Marc11} and cloud types~\cite{Burckhardt12,BernsteinB16} resolve conflicts automatically using convergent operations, but they impose important restrictions on data structures. Mergeable types~\cite{GKaki}, Cassandra~\cite{LakshmanM10}, CaCOPS~\cite{Lloyd11a}, Eiger~\cite{LloydFKA13}, and ChainReaction~\cite{AlmeidaLR13} use the last write wins strategy to ensure availability; however, they may lose data if concurrent writes happen frequently enough. Riak~\cite{BrownCME14} provides the ability to resolve write conflicts on the application level.

\section{Conclusion}\label{sec:conclusion}

We introduced the observable atomic consistency model which enables a
new extension of CvRDTs with non-monotonic operations. While lifting
a significant limitation of CvRDTs, we believe that it can significantly
simplify programming with CvRDTs. We presented the proof of state
convergence for systems providing observable atomic consistency. As a
next step, we discussed a new consistency protocol, called observable
atomic consistency protocol (OACP), which guarantees observable atomic
consistency for distributed systems. Experimental results show that
OACP is able to reduce the number of exchanged protocol messages
compared to the closely-related GSP protocol in several
microbenchmarks. This suggests that OACP can provide higher throughput
than GSP in some cases.



\bibliography{egsp}

\begin{thebibliography}{10}

\bibitem{agha86}
Gul~A. Agha.
\newblock {\em {ACTORS}: {A} Model of Concurrent Computation in Distributed
  Systems}.
\newblock Series in Artificial Intelligence. The MIT Press, 1986.

\bibitem{AlmeidaLR13}
S{\'e}rgio Almeida, Jo{\~a}o Leit{\~a}o, and Lu{\'i}s E.~T. Rodrigues.
\newblock {ChainReaction}: a causal+ consistent datastore based on chain
  replication.
\newblock In {\em EuroSys}, pages 85--98, 2013.

\bibitem{Balegas15}
Valter Balegas, S{\'e}rgio Duarte, Carla Ferreira, Rodrigo Rodrigues, Nuno~M.
  Pregui{\c c}a, Mahsa Najafzadeh, and Marc Shapiro.
\newblock Putting consistency back into eventual consistency.
\newblock In {\em EuroSys}, pages 6:1--6:16. ACM, 2015.

\bibitem{BaqueroAL16}
Carlos Baquero, Paulo~S{\'e}rgio Almeida, and Carl Lerche.
\newblock The problem with embedded crdt counters and a solution.
\newblock In {\em EuroSys}, pages 10:1--10:3. ACM, 2016.

\bibitem{RiakDTSource}
{Basho Technologies, Inc.}
\newblock Riak dt source code repository.
\newblock \url{https://github.com/basho/riak_dt}, 2012-2017.

\bibitem{orleans}
Phil Bernstein, Sergey Bykov, Alan Geller, Gabriel Kliot, and Jorgen Thelin.
\newblock Orleans: Distributed virtual actors for programmability and
  scalability.
\newblock Technical report, March 2014.

\bibitem{BernsteinBBCFKK17}
Philip~A. Bernstein, Sebastian Burckhardt, Sergey Bykov, Natacha Crooks,
  Jose~M. Faleiro, Gabriel Kliot, Alok Kumbhare, Muntasir~Raihan Rahman,
  Vivek~Shah 0001, Adriana Szekeres, and Jorgen Thelin.
\newblock Geo-distribution of actor-based services.
\newblock {\em PACMPL}, 1(OOPSLA):107:1--107:26, 2017.

\bibitem{BernsteinB16}
Philip~A. Bernstein and Sergey Bykov.
\newblock Developing cloud services using the orleans virtual actor model.
\newblock {\em IEEE Internet Computing}, pages 71--75, 2016.

\bibitem{BrownCME14}
Russell Brown, Sean Cribbs, Christopher Meiklejohn, and Sam Elliott.
\newblock {Riak DT} map: a composable, convergent replicated dictionary.
\newblock In {\em EuroSys}, page 1:1, 2014.

\bibitem{Burckhardt12}
Sebastian Burckhardt, Manuel F{\"a}hndrich, Daan Leijen, and Benjamin~P. Wood.
\newblock Cloud types for eventual consistency.
\newblock In {\em ECOOP}, pages 283--307. Springer, 2012.

\bibitem{Burckhardt13}
Sebastian Burckhardt, Alexey Gotsman, and Hongseok Yang.
\newblock Understanding eventual consistency.
\newblock Technical report, Microsoft Research (MSR), 2013.

\bibitem{BurckhardtLPF15}
Sebastian Burckhardt, Daan Leijen, Jonathan Protzenko, and Manuel
  F{\"a}hndrich.
\newblock Global sequence protocol: A robust abstraction for replicated shared
  state.
\newblock In {\em ECOOP}, pages 568--590. Schloss Dagstuhl - Leibniz-Zentrum
  fuer Informatik, 2015.

\bibitem{DefSchUrb04}
Defago, Schiper, and Urban.
\newblock Total order broadcast and multicast algorithms: Taxonomy and survey.
\newblock {\em CSURV: Computing Surveys}, 2004.

\bibitem{GilbertL02}
Seth Gilbert and Nancy~A. Lynch.
\newblock Brewer's conjecture and the feasibility of consistent, available,
  partition-tolerant web services.
\newblock {\em SIGACT News}, pages 51--59, 2002.

\bibitem{goodman89}
J.~R. Goodman.
\newblock Cache consistency and sequential consistency.
\newblock Technical report, IEEE Scalable Coherence Interface Working Group,
  March 1989.

\bibitem{Guerraoui16}
Rachid Guerraoui, Matej Pavlovic, and Dragos-Adrian Seredinschi.
\newblock Incremental consistency guarantees for replicated objects.
\newblock {\em CoRR}, 2016.

\bibitem{Haller12}
Philipp Haller.
\newblock On the integration of the actor model in mainstream technologies:
  {The} {Scala} perspective.
\newblock In {\em AGERE!@SPLASH}, pages 1--6. ACM, 2012.

\bibitem{HuntKJR10}
Patrick Hunt, Mahadev Konar, Flavio~Paiva Junqueira, and Benjamin Reed.
\newblock Zookeeper: Wait-free coordination for internet-scale systems.
\newblock In {\em USENIX Annual Technical Conference}, 2010.

\bibitem{JunqueiraRS11}
Flavio~Paiva Junqueira, Benjamin~C. Reed, and Marco Serafini.
\newblock Zab: High-performance broadcast for primary-backup systems.
\newblock pages 245--256. IEEE Compute Society, 2011.

\bibitem{GKaki}
Gowtham Kaki, KC~Sivaramakrishnan, Samodya Abeysiriwardane, and Suresh
  Jagannathan.
\newblock Mergeable types.
\newblock In {\em ML workshop}, 2017.

\bibitem{KiczalesHHKPG01}
Gregor Kiczales, Erik Hilsdale, Jim Hugunin, Mik Kersten, Jeffrey Palm, and
  William~G. Griswold.
\newblock An overview of aspectj.
\newblock In {\em ECOOP}, pages 327--353, 2001.

\bibitem{Kraska09}
Tim Kraska, Martin Hentschel, Gustavo Alonso, and Donald Kossmann.
\newblock Consistency rationing in the cloud: Pay only when it matters.
\newblock {\em PVLDB}, pages 253--264, 2009.

\bibitem{LakshmanM10}
Avinash Lakshman and Prashant Malik.
\newblock Cassandra: a decentralized structured storage system.
\newblock {\em Operating Systems Review}, pages 35--40, 2010.

\bibitem{Lamport05}
Leslie Lamport.
\newblock Fast paxos.
\newblock Technical report, Microsoft Research (MSR), July 2005.
\newblock URL: \url{ftp://ftp.research.microsoft.com/pub/tr/TR-2005-112.pdf}.

\bibitem{LCPRV14}
Cheng Li, Jo{\~a}o Leit{\~a}o, Allen Clement, Nuno~M. Pregui{\c c}a, Rodrigo
  Rodrigues, and Viktor Vafeiadis.
\newblock Automating the choice of consistency levels in replicated systems.
\newblock In {\em USENIX Annual Technical Conference}, pages 281--292, 2014.

\bibitem{LiPCGPR12}
Cheng Li, Daniel Porto, Allen Clement, Johannes Gehrke, Nuno~M. Pregui{\c c}a,
  and Rodrigo Rodrigues.
\newblock Making geo-replicated systems fast as possible, consistent when
  necessary.
\newblock In {\em OSDI}, pages 265--278. USENIX Association, 2012.

\bibitem{Li04a}
Jinyuan Li, Maxwell~N. Krohn, David Mazieres, and Dennis Shasha.
\newblock Secure untrusted data repository ({SUNDR}).
\newblock In {\em OSDI}, pages 121--136, December 2004.

\bibitem{Akka}
{Lightbend, Inc.}
\newblock {Akka}.
\newblock \url{http://akka.io/}, 2009.
\newblock Accessed: 2016-03-20.

\bibitem{Lloyd11a}
Wyatt Lloyd, Michael~J. Freedman, Michael Kaminsky, and David~G. Andersen.
\newblock Don't settle for eventual: Scalable causal consistency for wide-area
  storage with {COPS}.
\newblock In {\em SOSP}, pages 401--416. ACM Press, October 2011.

\bibitem{LloydFKA13}
Wyatt Lloyd, Michael~J. Freedman, Michael Kaminsky, and David~G. Andersen.
\newblock Stronger semantics for low-latency geo-replicated storage.
\newblock In {\em NSDI}, 2013.

\bibitem{MazSha02}
Mazieres and Shasha.
\newblock Building secure file systems out of byzantine storage.
\newblock In {\em PODC}, 2002.

\bibitem{MelgrattiR16}
Hern{\'a}n~C. Melgratti and Christian Rold{\'a}n.
\newblock A formal analysis of the global sequence protocol.
\newblock In {\em COORDINATION}, pages 175--191. Springer, 2016.

\bibitem{OngaroO14}
Diego Ongaro and John~K. Ousterhout.
\newblock In search of an understandable consensus algorithm.
\newblock In {\em USENIX ATC}, pages 305--319. USENIX Association, 2014.

\bibitem{Perkins15}
Dorian Perkins, Nitin Agrawal, Akshat Aranya, Curtis Yu, Younghwan Go,
  Harsha~V. Madhyastha, and Cristian Ungureanu.
\newblock Simba: tunable end-to-end data consistency for mobile apps.
\newblock In {\em EuroSys}, pages 7:1--7:16. ACM, 2015.

\bibitem{Marc11}
Marc Shapiro, Nuno Pregui{\c c}a, Carlos Baquero, and Marek Zawirski.
\newblock A comprehensive study of convergent and commutative replicated data
  types.
\newblock Technical Report RR-7506; inria-00555588, HAL CCSD, January 2011.

\bibitem{MarcS11}
Marc Shapiro, Nuno~M. Pregui{\c{c}}a, Carlos Baquero, and Marek Zawirski.
\newblock Conflict-free replicated data types.
\newblock In {\em {SSS}}, pages 386--400, 2011.

\bibitem{Sivaramakrishnan15}
K.~C. Sivaramakrishnan, Gowtham Kaki, and Suresh Jagannathan.
\newblock Declarative programming over eventually consistent data stores.
\newblock In {\em PLDI}, pages 413--424, 2015.

\bibitem{TerDB}
D.~B. Terry, M.~M. Theimer, K.~Petersen, A.~J. Demers, M.~J. Spreitzer, and
  C.~Hauser.
\newblock Managing update conflicts in {Bayou}, a weakly connected replicated
  storage system.
\newblock In {\em SOSP}, pages 172--183, December 1995.

\bibitem{Twissandra}
{Twissandra}.
\newblock {Twitter clone on Cassandra}.
\newblock \url{http://twissandra.com/}, 2014.
\newblock Accessed: 2018-02-26.

\bibitem{ZhaoH18}
Xin Zhao and Philipp Haller.
\newblock Observable atomic consistency for {CvRDTs}.
\newblock {\em CoRR}, abs/1802.09462, 2018.
\newblock URL: \url{https://arxiv.org/abs/1802.09462}, \href
  {http://arxiv.org/abs/1802.09462} {\path{arXiv:1802.09462}}.

\end{thebibliography}

\end{document}